\documentclass{LMCS}

\def\doi{9(3:6)2013}
\lmcsheading%
{\doi}
{1--31}
{}
{}
{Feb.~25, 2012}
{Aug.~28, 2013}
{}

\usepackage[T1]{fontenc}

\usepackage{ucs} 
\usepackage[utf8x]{inputenc} 
\usepackage{amsmath} 
\usepackage{amssymb} 
\usepackage{mathrsfs} 
\usepackage[all]{xy} 
\usepackage[normalem]{ulem} 
\usepackage{ bbold } 
\usepackage{hyperref,enumerate}

\newdir{ >}{{}*!/-10pt/@{>}} % mono helper
\newcommand\mto{\rightarrowtail}

\newcommand{\ra}{\rightarrow}
\newcommand{\Set}{\mbox{Set}}
\newcommand{\Fam}{\mbox{Fam}}
\newcommand{\hash}{\dagger}
\newcommand{\mcal}[1]{\mathcal{#1}}
\newcommand{\C}{\mcal{C}}
\newcommand{\cech}[1]{\check{#1}}

\newcommand{\proofrule}[2]{\ensuremath{
  \begin{array}[t]{@{}c@{}}
    \begin{array}[t]{lllll} 
      #1
    \end{array}\\
  \hline 
     #2
  \end{array}}}

\newcommand\pbc[1][dr]{\save*!/#1-1.8pc/#1:(-1,1)@^{|-}\restore}
\newcommand\lpbc[1][dr]{\save*!/#1-1.2pc/#1:(-1,1)@^{|-}\restore}

\newcommand\E{{\mathcal E}}
\newcommand\B{{\mathcal B}}
\newcommand\A{{\mathcal A}}
\newcommand\D{{\mathcal D}}

\newcommand\Ps{\mathscr{P}}
\newcommand\Pfi{\Ps\!\!{}_{\mathit{fin}}}

\newcommand\ascl{\mathrm{ASub}(\cl)}

\newcommand\Lam{\mathit{Lam}}
\newcommand\Fin{\mathit{Fin}}
\newcommand\Nat{\mathit{Nat}}
\newcommand\cl{\mathit{CL}}

\newcommand\alg{\mathit{Alg}}
\newcommand\dalg{\mbox{-}\!\alg}
\newcommand\coalg{\mathit{CoAlg}}
\newcommand\dcoalg{\mbox{-}\!\coalg}

\newcommand\ti{\!\times\!}

\theoremstyle{plain}
\newtheorem{theorem}{Theorem}[section]
\newtheorem{corollary}[theorem]{Corollary}
\newtheorem{lemma}[theorem]{Lemma}
\theoremstyle{remark}

\theoremstyle{definition}
\newtheorem{definition}[theorem]{Definition}
\newtheorem{example}[theorem]{Example}

\title{Indexed Induction and Coinduction,
  Fibrationally\vspace*{-0.15in}} 

\author[N.~Ghani]{Neil Ghani\rsuper a} \address{{\lsuper{a,c}}University of Strathclyde,
  Glasgow G1 1XH, UK} \email{\{neil.ghani, clement.fumex\}@strath.ac.uk}

\author[P.~Johann]{Patricia Johann\rsuper b} \address{{\lsuper b}Appalachian State University, Boone,
  NC 28608, USA} \email{johannp@cs.appstate.edu}

\author[C.~Fumex]{Cl\'ement Fumex\rsuper c} 
 \address{\vskip-6 pt}%{University of Strathclyde, Glasgow G1 1XH, UK} %
\begin{document}

\begin{abstract}
  This paper extends the fibrational approach to induction and
  coinduction pioneered by Hermida and Jacobs, and developed by the
  current authors, in two key directions. First, we present a dual to
  the sound induction rule for inductive types that we developed
  previously. That is, we present a sound coinduction rule for any
  data type arising as the carrier of the final coalgebra of a
  functor, thus relaxing Hermida and Jacobs' restriction to polynomial
  functors. To achieve this we introduce the notion of a
  \emph{quotient category with equality} (QCE) that i) abstracts the
  standard notion of a fibration of relations constructed from a given
  fibration; and ii) plays a role in the theory of coinduction dual to
  that played by a comprehension category with unit (CCU) in the
  theory of induction. Secondly, we show that inductive and
  coinductive indexed types also admit sound induction and coinduction
  rules. Indexed data types often arise as carriers of initial
  algebras and final coalgebras of functors on slice categories, so we
  give sufficient conditions under which we can construct, from a CCU
  (QCE) $U:\E \ra \B$, a fibration with base $\B/I$ that models
  indexing by $I$ and is also a CCU (resp., QCE). We finish the paper
  by considering the more general case of sound induction and
  coinduction rules for indexed data types when the indexing is
  itself given by a fibration.
\end{abstract}

\ACMCCS{[{\bf Theory of computation}]: Semantics of reasoning---Program semantics---Categorical semantics}

\maketitle

\section{Introduction}

Iteration operators provide a uniform way to express common and
naturally occurring patterns of recursion over inductive 
types. Expressing recursion via iteration operators makes code easier
to read, write, and understand; facilitates code reuse; guarantees
properties of programs such as totality and termination; and supports
optimising program transformations such as fold fusion and short cut
fusion.  Categorically, iteration operators arise from the initial
algebra semantics of data types: the constructors of an inductive 
type are modelled as a functor $F$, the data type itself is modelled
as the carrier $\mu F$ of the initial $F$-algebra $\mathit{in} : F(\mu
F) \ra \mu F$, and the iteration operator $\mathit{fold} : (F A \ra A)
\ra \mu F \ra A$ for $\mu F$ is the map sending each $F$-algebra $h :
FA \ra A$ to the unique $F$-algebra morphism from $\mathit{in}$ to
$h$.

Initial algebra semantics therefore provides a comprehensive theory of
iteration that is i) {\em principled}, in that it ensures that
programs have rigorous mathematical foundations that can be used to
give them meaning; ii) {\em expressive}, in that it is applicable to
{\em all} inductive types --- i.e., all types that are carriers of
initial algebras --- rather than just to syntactically defined classes
of data types such as polynomial ones; and iii) {\em sound}, in that
it is valid in any model --- set-theoretic, domain-theoretic,
realisability, etc. --- in which data types are interpreted as the
carriers of initial algebras.

Final coalgebra semantics gives an equally comprehensive understanding
of coinductive types. The destructors of a coinductive type are
modelled as a functor $F$, the data type itself is modelled as the
carrier $\nu F$ of the final $F$-coalgebra $\mathit{out} : \nu F \ra
F(\nu F)$, and the coiteration operator $\mathit{unfold} : (A \ra F A)
\ra A \ra \nu F$ for $\nu F$ is the map sending each $F$-coalgebra
$k:A \ra FA$ to the unique $F$-coalgebra morphism from $k$ to
$\mathit{out}$.  Final coalgebra semantics thus provides a theory of
coiteration that is as principled, expressive, and sound as that for
induction.

Since induction and iteration are closely linked, we might expect
initial algebra semantics to give a principled, expressive, and sound
theory of induction as well. But most theories of induction for a data
type $\mu F$, where $F:\B \ra \B$, are sound only under significant
restrictions on the category $\B$, the functor $F$, or the form and
nature of the property to be established. Recently, however, a
conceptual breakthrough in the theory of induction was made by Hermida
and Jacobs~\cite{hj98}. They first showed how to lift an arbitrary
functor $F$ on a base category $\B$ of types to a functor $\hat{F}$ on
a category of properties over those types. Then, taking the premises
of an induction rule for $\mu F$ to be an $\hat{F}$-algebra, their
main theorem shows that such a rule is sound if the lifting $\hat{F}$
preserves truth predicates. Hermida and Jacobs work in a fibrational,
and hence axiomatic, setting and treat {\em any} notion of property
that can be suitably fibred over $\B$. Moreover, they place no
stringent requirements on $\B$. Thus, they overcome two of the
aforementioned limitations. But since they give sound induction rules
only for polynomial data types, the limitation on the functors treated
remains in their work. The current authors~\cite{gjf10} subsequently
removed this final restriction to give sound induction rules for all
inductive types under conditions commensurate with those
in~\cite{hj98}.

In this paper, we extend the existing body of work in three key
directions. First, Hermida and Jacobs developed a fibrational theory
of coinduction to complement their theory of induction. But this
theory, too, is sound only for polynomial data types, and so does not
apply to final coalgebras of some key functors, such as the finite
powerset functor. In this paper, we derive a sound fibrational
coinduction rule for {\em every} coinductive type. Secondly, data
types arising as initial algebras of functors are fairly simple. More
sophisticated data types --- e.g., untyped lambda terms and red-black
trees --- are often modelled as inductive indexed types arising as
initial algebras of functors on slice categories, presheaf categories,
and similar structures. In this paper, we derive sound induction rules
for such inductive indexed types.  We do this by considering first the
special case of indexing via slice categories, and then the general
case where indexing is itself given by a suitable fibration.  Finally,
since we can derive sound induction rules for inductive types and
inductive indexed types, and sound coinduction rules for coinductive
types, we might expect to be able to derive sound coinduction rules
for coinductive indexed types, too. In this paper, we confirm that
this is the case and, again, consider first the special case of
indexing via slice categories and then the general situation.

We now describe the structure of the rest of this paper. After
describing the results in each section, we give a concrete example of
a widely-used data type and a corresponding logic for which the
results of that section can derive a sound induction or coinduction
rule, as appropriate, but for which such a rule cannot be derived from
previously known techniques of comparable generality. We thus show
that our framework not only facilitates an abstract conceptualisation
that reveals the essence of induction and coinduction, but also
significantly advances the state-of-the-art by being instantiable to a
larger class of data types and logics than ever before.  The rest of
this paper is structured as follows.

\begin{iteMize}{$\bullet$}

\item In Section 2, we recall the fibrational approach to induction
  pioneered in~\cite{hj98} and extended in~\cite{gjf10}. We also
  present a number of fibrations, each of which captures a different
  logic of interest. Finally, we recall conditions under which the
  fibrational induction rule we derive in~\cite{gjf10} can be
  instantiated to give a sound concrete induction rule for {\em any}
  inductive data type with respect to {\em any} such logic.

\item In Section~\ref{sec:coind} we extend the fibrational approach to
  coinduction from~\cite{hj98} to derive a coinduction rule that can
  be instantiated to give a sound concrete coinduction rule for {\em
    any} coinductive data type. We illustrate this by deriving a sound
  coinduction rule for the coinductive data type determined by the
  finite powerset functor. This functor is fundamental in the theory
  of bisimulation and labelled transition systems, but it is not a
  polynomial functor and so cannot be handled using the techniques of
  Hermida and Jacobs.

\item In Section~\ref{sec:indind} we use slice categories to model
  indexing of data types, and thus to give sound concrete induction
  rules for {\em all} inductive indexed data types. We apply this
  result to derive a sound induction rule for inductive type
  determined by indexed containers with respect to the families
  fibration, and then further specialise this rule to the inductive
  indexed data type of untyped lambda terms. The data type of untyped
  lambda terms is not determined by a polynomial functor, so the sound
  induction rule we derive for it is not simply an instantiation of
  Hermida and Jacobs' results.

\item In Section~\ref{sec:indcoind} we use slice categories again,
  this time to give sound concrete coinduction rules for {\em all}
  coinductive data types. We apply our results to derive sound
  coinduction (i.e., bisimulation) rules for coinductive types
  determined by indexed containers. These coinductive types are
  equivalent to Hancock and Hyvernat's interaction
  structures~\cite{hh06}. However, since they are not determined by
  polynomial functors, the coinduction rules we derive for them are
  not simply instantiations of Hermida and Jacobs' results.

\item In Section~\ref{sec:fibind} we study fibrational indexed
  induction by generalising the indexing of data types from slice
  categories to fibrations. We derive an induction rule that extends
  the one in Section~\ref{sec:indind} and show how it can be
  instantiated to give sound induction rules for set-indexed data
  types. Set-indexing occurs, for example, in mutually recursive
  definitions of data types.

\item In Section~\ref{sec:fibcoind} we similarly study fibrational
  indexed coinduction, derive a coinduction rule that extends the one
  in Section~\ref{sec:indcoind}, and point out that this rule can be
  instantiated to give sound coinduction rules for set-indexed data
  types.

\item In Section~\ref{sec:conc} we summarise our conclusions and
  discuss related work and possibilities for future research.
\end{iteMize}

\noindent This paper is a revised and expanded version of the conference
paper~\cite{fgj11}. Whereas the conference paper covers only indexing
modelled by slice categories, this paper also treats general
indexing. Accordingly, the material in Sections~\ref{sec:fibind}
and~\ref{sec:fibcoind} is entirely new.

\section{Induction in a Fibrational Setting}\label{sec:ind}

Fibrations support a uniform axiomatic approach to induction and
coinduction that is widely applicable and abstracts over the specific
choices of the category in which types are interpreted, the functor on
that category giving rise to the data type whose rules are to be
constructed, and the predicate those rules may be used to establish.
This is advantageous because i) the semantics of data types in
languages involving recursion and other effects usually involves
categories other than $\Set$; ii) in such circumstances, the standard
set-based interpretations of predicates are no longer germane; iii) in
any setting, there can be more than one reasonable notion of
predicate; and iv) fibrations allow induction and coinduction rules
for many classes of data types to be obtained by the instantiation of
a single generic theory, rather than developed on an {\em ad hoc}
basis. The genericness supported by fibrations provides a predictive
power that is the hallmark of any good scientific theory.

\subsection{Fibrations in a Nutshell}\label{sec:nutshell} 

We begin with fibrations. More details can be found in,
e.g.,~\cite{jac99,pav90}.
\begin{definition}
  Let $U : \E \ra \B$ be a functor. A morphism $g : Q \ra P$ in $\E$
  is {\em cartesian} above a morphism $f : X \ra Y$ in $\B$ if\, $Ug =
  f$ and, for every $g' : Q' \ra P$ in $\E$ with $Ug' = f v$ for
  some $v : UQ' \ra X$, there exists a unique $h : Q' \ra Q$ in $\E$
  such that $Uh = v$ and $g h = g'$.
\end{definition}

\noindent
A fibration is simply a functor $U:\E \ra \B$ that guarantees a large
supply of cartesian morphisms. The exact definition is as follows:

\begin{definition}
  Let $U : \E \ra \B$ be a functor. Then $U$ is a {\em fibration} if
  for every object $P$ of $\E$ and every morphism $f : X \ra UP$ in
  $\B$, there is a cartesian morphism above $f$ with codomain $P$. 
\end{definition}
If $U : \E \to \B$ is a fibration, we call $\B$ the {\em base
  category} of $U$ and $\E$ its {\em total category}. Objects of $\E$
are thought of as properties, objects of $\B$ are thought of as types,
and $U$ is thought to map each property $P$ in $\E$ to the type $UP$
about which it is a property. An object $P$ in $\E$ is said to be
\emph{above} its image $UP$ under $U$, and similarly for
morphisms. For any object $X$ of $\B$, we write $\E_X$ for the {\em
  fibre above $X$}, i.e., for the subcategory of $\E$ comprising
objects above $X$ and morphisms above the identity morphism $id_X$ on
$X$. Morphisms within a fibre are said to be {\em vertical}.

If $U:\E \ra \B$ is a fibration, $P$ is an object of $\E$, and $f : X
\ra UP$, we write $f^\S_P$ for the cartesian morphism above $f$ with
codomain $P$.  We omit the subscript $P$ when it can be inferred from
context. As with all entities defined via universal properties,
$f^\S_P$ is defined up to isomorphism; we write $f^*P$ for the domain
of $f^\S_P$. If $f : X \ra Y$ is a morphism in the base of a
fibration, then the function mapping each object $P$ of $\E_Y$ to
$f^*P$ extends to a functor $f^* : \E_Y \to \E_X$ called the {\em
  reindexing functor induced by $f$}. If we think of $f$ as performing
type-level substitution, then $f^*$ can be thought of as lifting $f$
to perform substitution of types into predicates.

\begin{example}\label{ex:famfib}
  The category $\Fam(\Set)$ has as objects pairs $(X, P)$ with $X$ a
  set and $P : X \ra \Set$. We call $X$ the {\em domain} of $(X,P)$
  and write $P$ for $(X,P)$ when convenient. A morphism from $P : X
  \ra \Set$ to $P' : X' \ra \Set$ is a pair $(f,f^\sim)$ of functions
  $f:X \rightarrow X'$ and $f^\sim:\forall x : X.\, P \,x \rightarrow
  P' (f \,x)$. The functor $U:\Fam(\Set) \ra \Set$ mapping $(X, P)$ to
  $X$ is called the {\em families fibration}. Here, the cartesian
  morphism associated with the object $P:Y \ra \Set$ in $\Fam(\Set)$
  and the morphism $f:X \ra Y$ in $\Set$ is the morphism $(f, id)$
  in $\Fam(\Set)$ from $P f : X \ra \Set$ to $P$.

\end{example}
\begin{example}\label{ex:cod}
  The {\em arrow category} of $\B$, denoted $\B^\to$, has morphisms of
  $\B$ as its objects. A morphism from $f:X\to Y$ to $f':X'\to Y'$ in
  $\B^\to$ is a pair $(\alpha_1, \alpha_2)$ of morphisms in $\B$ such
  that the following diagram commutes:
\[ \xymatrix{ X \;\; \ar[r]^{\alpha_1} \ar[d]_{f} &
    X' \ar[d]^{f'} \\
    Y \ar[r]_{\alpha_2} & Y'}
\]
\noindent
  The codomain functor $\mathit{cod} : \B^{\ra} \ra \B$ maps an object
  $f:X \ra Y$ of $\B^{\ra}$ to the object $Y$ of $\B$. If $\B$ has
  pullbacks, then $\mathit{cod}$ is a fibration, called the {\em
    codomain fibration over $\B$}. Indeed, given an object $f : X \to
  Y$ in the fibre above $Y$ and a morphism $f':X' \ra Y$ in $\B$, the
  pullback of $f$ along $f'$ gives the cartesian morphism above $f'$.
  Similarly, the domain functor $\mathit{dom} : \B^{\ra} \ra \B$ is a
  fibration, called the {\em domain fibration over $\B$}. No
  conditions on $\B$ are required.
\end{example}

A useful restriction of the previous fibration considers (equivalence
classes of) monic maps only:

\begin{example} 
\label{ex:subfib}
Let $\B$ be a category with pullbacks. Let $Sub(\B)$ be the category
of subobjects of $\B$, i.e., let the objects of $Sub(\B)$ be
equivalence classes of monos (where $m : X \rightarrow I$ and $n : X
\rightarrow I$ are equivalent iff they are isomorphic in the slice
category $\B/I$). The {\em subobject fibration over $\B$} is the
fibration $U : Sub(\B) \rightarrow \B$ that sends an equivalence class
$[m]$ to the codomain of $m$. Reindexing is well-defined because the
pullback of a mono along any morphism is again a mono. Note that every
fibre $Sub(\B)_I$ is a preorder, and thus that $U$ is a fibred
preorder. Fibred preorders can be thought of logically as modelling
just provability, rather than proofs themselves.
\end{example}

The following fibration appears as Example 4.8.7 (iii) in~\cite{jac99}:

\begin{example}\label{ex:CL}
  Let $CL$ be the category of complete lattices with functions
  preserving all joins between them. If $X$ is a complete lattice,
  then a subset $A \subseteq X$ is {\em admissible} if $A$ is closed
  under joins in $X$. We write $ASub(CL)$ for the category whose
  objects are pairs $(X,A)$, where $X$ is a complete lattice and $A$
  is an admissible subset of $X$, and whose morphisms from $(X,A)$ to
  $(Y,B)$ are morphisms $f : X \rightarrow Y$ in $CL$ such that $x \in
  A$ implies $f(x) \in B$.  Admissible subsets of complete lattices
  form a fibration $U : ASub(CL) \rightarrow CL$. Indeed, if $(Y,B)$
  is an object in $ASub(CL)$, if $f : X \to Y$ in $CL$, and if we
  define $f^*(Y,B) = (X, \{x \in X \mid f(x) \in B\})$, then
  $f^*(Y,B)$ is actually an object of $ASub(CL)$ since $f$ preserves
  joins. Moreover, a cartesian morphism $f^\S : f^*(Y,B) \rightarrow
  (Y,B)$ is given by $f$ itself.
\end{example}

\subsubsection{Bifibrations}

We will later need the generalisation of the notion of a fibration to
that of a bifibration. Since bifibrations are defined in terms of
opfibrations, we begin by defining these. Abstractly, $U : \E \to \B$
is an opfibration if $U^{op}:\E^{op} \ra \B^{op}$ is a fibration. This
characterisation has the merit of allowing us to use duality to
establish properties of opfibrations from properties of fibrations,
but a more concrete definition can be obtained by unwinding the
characterisation above.

\begin{definition}
  Let $U : \E \ra \B$ be a functor. A morphism $g : P \ra Q$ in $\E$
  is {\em opcartesian} above a morphism $f : X \ra Y$ in $\B$ if $Ug =
  f$ and, for every $g' : P \ra Q'$ in $\E$ with $Ug' = v f$ for
  some $v : Y \ra UQ'$, there exists a unique $h : Q \ra Q'$ in $\E$
  such that $Uh = v$ and $h g = g'$.
\end{definition}
\noindent
Just as a fibration is simply a functor that has a plentiful supply
of cartesian morphisms in its domain, so an opfibration is a functor
that has a plentiful supply of opcartesian morphisms in its
domain. We have:

\begin{definition}
  If $U : \E \ra \B$ is a functor, then $U$ is an {\em opfibration} if
  for every object $P$ of $\E$ and every morphism $f : UP \ra Y$ in
  $\B$ there is an opcartesian morphism in $\E$ above $f$ with domain
  $P$.  A functor $U$ is a {\em bifibration} if it is simultaneously a
  fibration and an opfibration.
\end{definition}
If $U$ is an opfibration, $P$ is an object of $\E$ and $f : UP \ra Y$
is a morphism of $\B$, then we denote the opcartesian morphism above
$f$ with domain $P$ by $f_\S^P$ and note that, as with cartesian
morphisms, this is defined up to isomorphism. We write $\Sigma_fP$ for
the codomain of $f_\S^P$ and omit the superscript $P$ when it can be
inferred from context.  If $f:X \ra Y$ is a morphism in the base of an
opfibration, then the function mapping each object $P$ of $\E_X$ to
$\Sigma_f P$ extends to a functor $\Sigma_f : \E_X \to \E_Y$ called
the {\em opreindexing functor induced by $f$}. The following useful
result is from~\cite{jac93}:
\begin{lemma}\label{lem:bifdef}
  Let $U: \E \ra \B$ be a fibration. Then $U$ is a bifibration iff,
  for every morphism $f:X \ra Y$ in $\B$, $f^*$ has a left adjoint
  $\Sigma_f$.
\end{lemma}

Both the families fibration and the codomain fibration are
opfibrations, and thus bifibrations. In the families fibration, if
$f:X \ra Y$ is a function, and $P:X \ra \Set$ is an object of
$\Fam(\Set)$ above $X$, then the associated opcartesian morphism has
as codomain the function from $Y$ to $\Set$ that maps $y \in Y$ to the
disjoint union $\biguplus_{\{x \in X | f x = y\}} P x$. The first
component of the opcartesian morphism is $f$, and its second component
maps $x \in X$ and $p \in Px$ to the pair $(x,p)$.  In the codomain
fibration, if $f:X \ra Y$ is a morphism in the base category and $g:Z
\ra X$ is above $X$, then we can construct the opcartesian morphism
consisting of the pair of morphisms $(id_Z, f)$ from $g$ to $fg$. In
general, the subobject fibration over $\B$ is not an opfibration, and
hence not a bifibration. However, as shown in Lemma~4.4.6
of~\cite{jac99}, if $\B$ is a regular category then it is.

%However, the concepts in the following definition allow us to describe
%a condition under which it is.

%\begin{definition}
%  Let $\B$ be a category. A morphism $u : I \to J$ in $\B$ factors
%  through the mono $v : K \to J$ if there exists a $w : I \to K$ such
%  that $u = vw$. If such a $w$ exists, then it is unique since $v$ is
%  a mono. The mono $v$ is the least mono through which $u$ factors if
%  i) $u$ factors as $vw$ and ii) for any mono $v' : K' \to J$ though
%  which $u$ factors as $v'w'$, there is a (necessarily unique)
%  morphism from $K$ to $K'$ such that
%  \[\xymatrix{I\ar@/^/[rr]^u\ar[rd]^w\ar[rdd]_{w'} && J\\
%      &K\ar@{>->}[ru]^v\ar@{.>}[d]\\&K'\ar[ruu]_{v'}}\] A category
%  $\B$ has \emph{images} if every morphism $u$ in $\B$ has a least
%  mono $m(u)$ through which it factors. We write $Im(u)$ for the
%  domain of $m(u)$. A category with images has \emph{stable images}
%  if its images are stable under pullback, i.e., if the diagram on the right
%  below is a pullback whenever the one on the left is:
%    \[\xymatrix{K\ar[d]_v\ar[r]&I\ar[d]^u\\L\ar[r]_w&J}\quad\quad\quad
%  \xymatrix{\mathrm{Im}(v)\ar@{.>}[r]\ar@{ >->}[d]&
%    \mathrm{Im}(u)\ar@{ >->}[d]\\L\ar[r]_w&J}\] We say that a category
%  is \emph{regular} if it has finite limits and stable images. 
%\end{definition}

%It is shown in~\cite{jac99} that the subobject fibration $U:
%Sub{\B}\to\B$ is a bifibration whenever $\B$ is a regular
%category. Indeed, for a morphism $u:I\to J$ in $\B$ and a mono
%$m:X\mto I$, define $\Sigma_um$ to be $\mathrm{m}(u\circ
%m):\mathrm{Im}(u\circ m)\mto J$.

\subsubsection{Beck-Chevalley Conditions and Fibred Adjunctions}    

Beck-Chevalley conditions are used to guarantee that reindexing
satisfies desirable commutativity properties. See~\cite{jac99} for an
expanded treatment of the following discussion.

\begin{definition}\label{def:beckcc}
  Let $U:\E\to\B$ be a bifibration. We say that $U$ \emph{satisfies
    the Beck-Chevalley condition (for opreindexing)} if for any
  pullback square
  \[ \xymatrix{ A \ar[r]^t \ar[d]_s \pbc& B \ar[d]^f \\ C \ar[r]_g &
    D} \] in $\B$, the canonical natural transformation $\Sigma_s t^*
  \xrightarrow{.} g^*\Sigma_f$ defined as \[\Sigma_s t^*
  \xrightarrow{\Sigma_st^* \eta^f} \Sigma_s t^* f^*\Sigma_f
  \xrightarrow{\cong} \Sigma_s s^* g^* \Sigma_f
  \xrightarrow{\epsilon^s g^*\Sigma_f} g^* \Sigma_f\] is an
  isomorphism. Here, $\eta^f$ is the unit of the adjunction
  $\Sigma_f\dashv f^*$ and $\epsilon^s$ is the counit of the
  adjunction $\Sigma_s\dashv s^*$.
\end{definition}
\noindent 
It is easy to check that the families fibration, the codomain
fibration, and the fibration of admissible subsets of complete
lattices satisfy the Beck-Chevalley condition. In addition, the
subobject fibration over $\B$ satisfies the Beck-Chevalley condition
if $\B$ is regular. In addition, we have:

\begin{lemma}\label{lem:bcmono}
  Let $U:\E\to\B$ be a bifibration that satisfies the Beck-Chevalley
  condition. % For any mono $f:X\to Y$ in $\B$ and any $P$ above $X$,
  % the unit $\eta : P\to f^*\Sigma_f P$ is an isomorphism. Or,
  % equivalently, 
  Then for any mono $f:X\to Y$ in $\B$,
  \begin{enumerate}[\em(1)]
  \item the functor $\Sigma_f:\E_X\to\E_Y$ is full and faithful, and
  \item any opcartesian morphism above $f$ is also cartesian.
  \end{enumerate}
\end{lemma}

Given that fibrations are the fundamental structures used in this
paper, it is natural to ask what morphisms between such structures
might be. In general, we can consider this question in a setting where
the fibrations can have different base categories. However, for our
purposes we only need consider the special case where the base
categories of the fibrations involved are the same. In this situation
we have the following definition:

\begin{definition}
  Let $\B$ be a category. Given two fibrations $U:\E\to\B$ and
  $U':\E'\to\B$ with base category $\B$, a {\em fibred functor} from
  $U$ to $U'$ above $\B$ is a functor $H:\E\to\E'$ such that $H$
  preserves cartesian morphisms and the following diagram commutes:
  \[\xymatrix{\E \ar[rd]_U \ar[rr]^H && \E' \ar[ld]^{U'}\\
    & \B}\]
\end{definition}

One of the key points about fibred functors is that they allow us to
define fibred adjunctions, and thus to lift standard categorical
structures to the fibred setting. In the special case when the base
categories of the fibrations are the same, a fibred adjunction is
defined as follows:

\begin{definition}
  Let $\B$ be a category and $U:\E \ra \B$ and $U':\E' \ra \B$ be
  fibrations. Given two fibred functors $G:U\to U'$ and $F:U'\to U$
  above $\B$, we say that $G$ is a {\em fibred right adjoint} of $F$
  above $\B$ iff $G$ is right adjoint to $F$ and the unit (or,
  equivalently, counit) of the adjunction $F\dashv G$ is vertical.  We
  say that the adjunction $F \dashv G$ is a {\em fibred adjunction}
  above $\B$.
\end{definition}
\noindent
Henceforth, we speak only of fibred functors and fibred adjunctions,
and leave implicit the fact they are above a particular category. 

The definition of a fibred adjunction can be given an alternative form
in terms of a collection of adjunctions between corresponding fibres
of fibrations and a coherence property linking these adjunctions
together. To see this, we first introduce the following helpful
notation.  Given fibrations $U:\E \ra \B$ and $U':\E' \ra \B$, a
fibred functor $F:U \ra U'$, and an object $X$ of $\B$, we denote by
$F_X:\E_X \ra \E'_X$ the restriction of $F$ to the fibre $\E_X$. We
know the image of $F_X$ lies within $\E'_X$ because $F$ is fibred. We
have:

\begin{lemma}\label{lem:fibadj}
  Let $\B$ be a category, let $U:\E\to\B$ and $U':\E'\to\B$ be
  fibrations, and let $G:U \ra U'$ be a fibred functor. Then $G$ has a
  fibred left adjoint iff the following two conditions hold:
\begin{enumerate}[\em(1)]
\item for any $X$ in $\B$, $G_X$ has a left adjoint $F_X$, and
\item %the Beck-Chevalley condition for fibred adjunctions holds, i.e.,
  for every morphism $f:X\to Y$ in $\B$ with associated reindexing
  functors $f^*$ and $f^{\dagger}$ with respect to $U$and $U'$,
  respectively, the canonical natural transformation from
  $F_Xf^{\dagger}$ to $f^* F_Y$ obtained as the transpose\footnote{If
    $F \dashv G: \C \rightarrow \D$, then the transpose of a morphism
    $f : FX \rightarrow Y$ is $(Gf)\eta_X$ and the transpose of a
    morphism $g : X \ra GY$ is $\epsilon_Y (Fg)$, where $\eta :
    \mathit{Id}_\C \ra GF$ and $\epsilon : FG \ra \mathit{Id}_\D$ are
    the unit and counit, respectively, of the adjunction $F \dashv
    G$.} of $\rho (f^{\dagger}\eta)$ is an isomorphism. Here, $\rho :
  f^{\dagger}G_YF_Y \rightarrow G_Xf^*F_Y$ arises from the fact that
  $G$, and hence $G_X$, preserves cartesian morphisms.
\end{enumerate}
\end{lemma}

\noindent
Suppose $X$ is in $\B$ and $P$ is in $\E'_X$ in the setting of
Lemma~\ref{lem:fibadj}. Then $FP = F_XP$, and if $\eta$ is the unit of
$F \dashv G$ and $\eta^X$ is the unit of $F_X \dashv G_X$, then
$\eta_P = (\eta^X)_P$.

We conclude this section with a lemma about (non-fibred) adjunctions
and the preservation of cartesian and opcartesian morphisms.
\begin{lemma}\label{lem:adjandcart}
  Let $U:\E\to\B$ and $U':\E'\to\B$ be fibrations. Further, let
  $F:\E\to\E'$ and $G:\E'\to\E$ be adjoint functors $F \dashv G$ with
  vertical unit (or equivalently, counit) such that $U = U' F$
  and $U' = U G$. Then the functor $F$ preserves opcartesian
  morphisms and the functor $G$ preserves cartesian morphisms.
\end{lemma}
\begin{proof}
  We prove only that $G$ preserves cartesian morphisms; the second
  result is then obtained by dualising. Let $f:X\to Y$ be a morphism
  in $\B$ and let $u:Q\to P$ be the cartesian morphism above $f$ in
  $\E'$. We will prove that $G u$ is cartesian above $f$ in $\E$. To
  do this, let $l:R\to G P$ be a morphism in $\E$ above $f g$ for some
  $g$ in $\B$. Then the transpose $ \epsilon_P (Fl)$ of $l$ is above $
  f g$ in $\E'$ because the counit $\epsilon$ of the adjunction $F
  \dashv G$ is vertical. We then have a unique morphism $v:FR\to Q$ in
  $\E'$ above $g$ such that $uv = \epsilon_P (F l)$ since $u$ is
  cartesian. Because $\eta$ is vertical, this gives us a unique
  morphism $ (Gv) \eta_R:R\to GQ$ in $\E$ above $g$ such that $(Gu)
  (Gv\eta_R) = l$.
\end{proof}

\subsection{Fibrational Induction in Another Nutshell}

At the heart of Hermida and Jacobs' approach to induction is the
observation that if $U:\E \ra \B$ is a fibration and $F:\B \ra \B$ is
a functor, then $F$ can be lifted to a functor $\hat{F}:\E \ra \E$ and
the premises of the induction rule for $\mu F$ can be taken to be an
$\hat{F}$-algebra. Hermida and Jacobs observed that, crucially, this
lifting must be truth-preserving. We define these terms now.

\begin{definition}\label{def:multi}
  Let $U : \E \rightarrow \B$ be a fibration and $F : \B \to \B$ be a
  functor.  A {\em lifting} of $F$ with respect to $U$ is a functor
  $\hat{F}:\E \rightarrow \E$ such that $U \hat{F} = F U$.  If each
  fibre $\E_X$ has a terminal object, and if reindexing preserves
  terminal objects, then we say that $U$ {\em has fibred terminal
    objects}. In this case, the map assigning to every $X$ in $\B$ the
  terminal object in $\E_X$ defines a full and faithful functor $K_U$
  that is called the {\em truth functor} for $U$ and is right adjoint
  to $U$. We omit the subscript on $K_U$ when this can be inferred. A
  lifting $\hat{F}$ of $F$ is said to be {\em truth-preserving} if $K
  F \cong \hat{F} K$.
\end{definition}

The families fibration has fibred terminal objects: the terminal
object in the fibre above $X$ is the function mapping each $x \in X$
to the one-element set. The codomain fibration $\mathit{cod}$ also has
fibred terminal objects: the terminal object in the fibre above $X$ is
$\mathit{id}_X$. The subobject fibration has fibred terminal objects:
the terminal object in the fibre above $X$ is the equivalence class of
$\mathit{id}_X$.  A truth-preserving lifting $F^{\to}$ of $F$ with
respect to $\mathit{cod}$ is given by the action of $F$ on morphisms.
Truth-preserving liftings of functors with respect to the families
fibration and the subobject fibration over a regular category can be
obtained from the results of this section.

As mentioned in the introduction, in the fibrational approach to
induction the premises of an induction rule for a data type $\mu F$
are taken to be an $\hat{F}$-algebra $\alpha:\hat{F}P \ra P$. But what
about the conclusion of such an induction rule? Since its premises are
an $\hat{F}$-algebra, it is reasonable to expect its conclusion to be
the unique mediating morphism from the initial $\hat{F}$-algebra to
$\alpha$.  But this expectation is thwarted because an initial
$\hat{F}$-algebra is not, in general, guaranteed to exist. We
therefore seek conditions ensuring that, for every functor $F$ on the
base category of a fibration $U$, its lifting $\hat{F}$ has an initial
algebra. Moreover, our examples below suggest that the carrier of this
initial $\hat{F}$-algebra should be $K(\mu F)$, where $K$ is the truth
functor for $U$.  Fortunately, we already know that any
truth-preserving lifting $\hat{F}$ of $F$ defines a functor $K\dalg_F
:\alg_F\to\alg_{\hat F}$ mapping an $F$-algebra $\alpha : FX
\rightarrow X$ to the $\hat F$-algebra $K\alpha : KFX \cong \hat F KX
\ra KX$.  Soundness of the induction rule thus turns out to be
equivalent to requiring that applying $K\dalg_F$ to the initial
$F$-algebra gives the initial $\hat{F}$-algebra. We capture this
discussion formally as follows:

\begin{definition}\label{def:ind}
  Let $U:\E\to\B$ be a fibration with truth functor $K:\B\to\E$ and
  let $F:\B\to\B$ be a functor whose initial algebra has carrier $\mu
  F$. We say that a truth-preserving lifting $\hat F$ of $F$
  \emph{defines a sound induction rule for $\mu F$ in $U$} if the
  functor $K\dalg_F :\alg_F\to\alg_{\hat F}$ preserves initial objects.
\end{definition}
\noindent
We will omit explicit reference to $U$ when it is clear from context.
In the situation of Definition~\ref{def:ind}, the generic fibrational
induction rule is given by
\[\mathit{ind}_F : (\forall P: \E_X).\, (\hat{F} P \ra P) \ra K(\mu F) \ra P \]
and its soundness ensures that if $\alpha:\hat{F} P \ra P$ is above
$f$, then $\mathit{ind}_F \,P\, \alpha$ is above $\mathit{fold} \, f$.

To see how the above categorical definition of an induction rule
corresponds to our intuitive understanding, we look at an example
before returning to the general discussion of fibrational induction.

\begin{example}
  The data type $\mathit{Nat}$ of natural numbers is $\mu N$, where
  $N$ is the functor on $\Set$ defined by $N\, X = 1 + X$. A lifting
  $\hat{N}$ of $N$ from $\Set$ to $\Fam(\Set)$ is given by
\[\begin{array}{lll}
\hat{N} P\, (\mathit{inl} \,*) & \;=\; & 1 \\
\hat{N} P \,(\mathit{inr} \,n) & \;=\; & P \,n  
\end{array}\]
An $\hat{N}$-algebra with carrier $P:\mathit{Nat} \rightarrow \Set$ can
be given by $\mathit{in} : 1 + \mathit{Nat} \rightarrow \mathit{Nat}$ and $in^\sim
\; : \; \forall t : 1 + \mathit{Nat}. \;\hat{N} P \,t \ra P (\mathit{in}
\,t)$.  Since $\mathit{in} \,(\mathit{inl} \, *) = 0$ and $\mathit{in}
\,(\mathit{inr} \, n) = n + 1$, we see that $in^\sim$ consists of an
element $h_1 : P \,0$ and a function $h_2 : \forall n :
\mathit{Nat}. \; P\, n \ra P\,(n+1)$.  These are exactly the premises
of the standard induction rule we learn on the playground. As for the
conclusion of the induction rule, we first note that $\mathit{fold} \,
in = id$, so that the induction rule has as its conclusion a morphism
of predicates from $K \mathit{Nat}$ to $P$ whose first component is
$\mathit{id}$. The second component will be a function with type
$\forall n : \mathit{Nat}. \,1 \ra P n$, i.e., a function that gives,
for $n \in \mathit{Nat}$, a proof in $P n$. This is exactly as expected.
\end{example}

Definition~\ref{def:ind} naturally leads us to ask for conditions on a
fibration $U$ guaranteeing that a truth-preserving lifting of a
functor $F$ defines a sound induction rule for $\mu F$. Hermida
and Jacobs' key theorem states that a sufficient condition is that $U$
be a comprehension category with unit. 

\begin{definition}
  A {\em comprehension category with unit} (CCU) is a fibration $U:\E
  \ra \B$ with a truth functor $K_U$ that has a right adjoint
  $\{-\}_U$. In this case, $\{-\}_U$ is called the {\em comprehension
    functor} for $U$.
\end{definition}
\noindent
We omit the subscript on $\{-\}_U$ when this can be inferred from
context. 

The families fibration is a CCU: the comprehension functor maps a
predicate $P: X \ra \Set$ to the set $\Sigma x\!:\!X. \; Px$.  The
fibration $\mathit{cod}$ is the canonical CCU: the comprehension
functor is the domain functor $\mathit{dom} : \B^\to \to \B$ mapping
$f :X \to Y$ in $\B^\to$ to $X$. The subobject fibration over a
category $\B$ is a CCU: the comprehension functor maps an equivalence
class to the domain of a (chosen) representative. As shown
in~\cite{hj98}, truth-preserving liftings for CCUs define sound
induction rules. That is,

\begin{theorem}\label{thm:ind}
  Let $U:\E\to\B$ be a CCU and $F:\B\to\B$ be a functor whose initial
  algebra has carrier $\mu F$. Then every truth-preserving lifting
  $\hat F$ of $F$ with respect to $U$ defines a sound induction rule
  for $\mu F$.
\end{theorem}

The proof of this theorem is conceptually simple: Hermida and Jacobs
show that under the assumptions of the theorem, each functor $K\dalg_F$
has a right adjoint and therefore preserves all colimits, including
the initial object. This very elegant theorem shows that fibrations
provide just the right structure to derive sound induction rules for
inductive types whose underlying functors have truth-preserving
liftings. And it's amazing to see such structure captured so smoothly
as the existence of a pair of adjoints to the fibration
itself. However, there is still one missing ingredient, namely, a set
of conditions under which functors are guaranteed to have
truth-preserving liftings. Hermida and Jacobs~\cite{hj98} provided
truth-preserving liftings, and thus sound induction rules, only for
polynomial functors. This situation was rectified in~\cite{gjf10},
where it was shown that every functor has a truth-preserving lifting
with respect to every CCU that is also a bifibration. Such CCUs are
called {\em Lawvere categories}.

\begin{definition}
  A fibration $U : \E \rightarrow \B$ is a {\em Lawvere category} if
  it is a CCU that is also a bifibration.
\end{definition}

If $\epsilon$ is the counit of the adjunction $K \dashv \{-\}$ for a
CCU $U$, then $\pi_P = U \epsilon_P$ defines a natural transformation
$\pi : \{-\} \rightarrow U$. 
%(The domain of $\pi_P$ really is $\{P\}$ since $UK = \mathit{Id}$.)
Moreover, $\pi$ extends to a functor $\pi:\E \rightarrow
\B^{\rightarrow}$ in the obvious way.

\begin{lemma}\label{lem:deflifting}
  Let $U:\E \rightarrow \B$ be a Lawvere category. Then $\pi$ has a
  left adjoint $I:\B^{\rightarrow} \rightarrow \E$ defined by $I \,(f
  : X \rightarrow Y)\, =\, \Sigma_f\, (K X)$.
\end{lemma}
\noindent
For any functor $F$, the composition $\hat{F} = I F^{\to} \pi : \E
\rightarrow \E$ defines a truth-preserving lifting with respect to the
Lawvere category $U$~\cite{gjf11}. Here, $F^{\to}$ is the lifting
given after Definition~\ref{def:multi} of $F$ to the total category of
the codomain fibration. Concretely, $\hat{F}P = \Sigma_{F \pi_P} K F
\{P\}$. Thus, if $U$ is a Lawvere category and $F$ has an initial
algebra $\mu F$, then Theorem~\ref{thm:ind} guarantees that $\hat F$
defines a sound induction rule for $\mu F$. Indeed, we have:

\begin{theorem}\label{thm:inda}
If $U:\E \to \B$ is a Lawvere category and $F : \B \to \B$ is a
functor whose initial algebra has carrier $\mu F$, then there exists a
sound induction rule for $\mu F$ in $U$.
\end{theorem}
\noindent
If $\B$ has pullbacks, so that the functor $\mathit{cod}$ is actually
a fibration, then the following diagram establishes that we have
actually given a uniform modular construction of a lifting with
respect to any Lawvere category by factorisation through the lifting
for $\mathit{cod}$:

\[\xymatrix{{\E}\ar[dr]_U \ar@/^/[rr]^\pi \ar@{}[rr]|\top & &
   \ar@/^/[ll]^I \ar[ld]^{cod} {\B^\to}\\ & {\B} & }\]

\section{Coinduction}\label{sec:coind}

In~\cite{hj98}, Hermida and Jacobs augmented their sound fibrational
induction rules for carriers of initial algebras of polynomial
functors with a sound coinduction rule for carriers of final
coalgebras of polynomial functors. The goals of this section are to
recall the results of Hermida and Jacobs, and to extend them to give
sound coinduction rules for carriers of final coalgebras of functors.

Hermida and Jacobs begin by observing that coinduction is concerned
with relations. Given a fibration $U$ whose total category is thought
of as a category of predicates, they therefore construct a new
fibration $Rel(U)$ whose total category is thought of as a category of
relations.

\begin{definition}\label{def:cob}
  Let $U:\E \ra \B$ be a fibration where $\B$ has products, and let
  $\Delta:\B \ra \B$ be the diagonal functor sending an object $X$ to
  $X \times X$. Then the fibration $Rel(U):Rel(\E) \ra \B$ is obtained
  by the pullback of $U$ along $\Delta$. We call $Rel(U)$ the {\em
    relations fibration} for $U$.
\end{definition}

That the pullback of a fibration along any functor is a fibration is
well-known~\cite{jac94}, and the process of pulling back a fibration
along a functor $F$ to obtain a new fibration is called {\em change of
  base} along $F$. Since an opfibration from $\E$ to $\B$ is a
fibration from $\E^{op}$ to $\B^{op}$, change of base preserves
opfibrations as well as fibrations, and therefore preserves
bifibrations. Below we denote the pullback of {\em any} functor $F:\A
\ra \B$ along a functor $G:\B' \ra \B$ by $G^*F: G^*\A \ra \B'$. The
objects of $G^*\A$ are pairs $(X,Y)$ such that $GX = FY$, and $G^*F$
maps the pair $(X,Y)$ to the object $X$.  We write $Y$ for $(X,Y)$ in
$G^*A$ when convenient.

If $U : \E \to \B$ is a bifibration, then change of base along a
natural transformation $\alpha : F \to G$ induces an adjunction
between $F^*\E$ and $G^*\E$. We have:

\begin{lemma}\label{lem:natadj}
  For $U:\E\to\B$ a bifibration and $\alpha:F\to G$ a natural
  transformation with $F,G:\A\to\B$. There is an adjunction
  \[\xymatrix{F^*\E \ar[rd]_{F^*U} \ar@/^/[rr]^{\Sigma_\alpha}
    \ar@{}[rr]|\bot&& \ar@/^/[ll]^{\alpha^*} G^*\E \ar[ld]^{G^*U}\\&\A}\]
  with
  \begin{align*}
    &\Sigma_\alpha (X,P) = (X,\,\Sigma_{\alpha_X} P)\\
    &\alpha^*(X,Q) = (X,\,(\alpha_X)^*Q)
  \end{align*}
%Here, $(\alpha_X)^*$ and $(\Sigma_{\alpha_X}$ denote the reindexing and
%opreindexing functors induced by $\alpha_X$, respectively.
Furthermore, if $U$ satisfies the Beck-Chevalley condition 
%from Definition~\ref{def:beckcc} 
and the components of $\alpha$ are monos, then $\Sigma_\alpha$ is full
and faithful.
\end{lemma}
\begin{proof}
  Straightforward from the definitions of $\Sigma_\alpha$ and
  $\alpha^*$ and Lemmas~\ref{lem:bifdef} and~\ref{lem:bcmono}.
\end{proof}

Definition~\ref{def:cob} entails that the fibre of $Rel(\E)$ above $X$
is the fibre $\E_{X \times X}$. A morphism from $(X,Y)$ to $(X',Y')$
in $Rel(\E)$ consists of a pair of morphisms $\alpha : X \to X'$ and
$\beta : Y \to Y'$ such that $U \beta = \alpha \times \alpha$. Change
of base is well-known to preserve fibred terminal
objects~\cite{her93b}. It therefore preserves truth functors, so that
$Rel(U)$ has a truth functor whenever $U$ does. This is given by
$K_{Rel(U)} X = K_U (X \times X)$.

\begin{example}
Let $U$ be the families fibration. Then the fibre of $Rel(U)$ above a
set $X$ consists of functions $R:X \times X \ra \Set$. These are, as
intended, just (set-valued) relations. The truth functor for $Rel(U)$
maps a set $X$ to the relation $R : X \times X \ra \Set$ that maps
each pair $(x,x')$ to the one-element set.
\end{example}

In the inductive setting, truth-preserving liftings were needed. In
the coinductive setting, we need equality-preserving liftings, where
the equality functor is defined as follows:

\begin{definition}\label{def:eq-fun}
  Let $U:\E \ra \B$ be a bifibration where $\B$ has products, and let
  $K$ be the truth functor for $U$. Let $\delta: \mathit{Id}_\B \ra
  \Delta$ be the diagonal natural transformation for $\Delta$ with
  components $\delta_X: X \ra X \times X$, and let $\Sigma_{\delta} :
  \E \ra Rel(\E)$ be the functor mapping an object $P$ above $X$ to
  the object $\Sigma_{\delta_X} P$. Note that $\Sigma_{\delta_X} P$ is
  above $X \times X$ in $\E$ and above $X$ in $Rel(\E)$. The {\em
    equality functor} for $U$ is the functor $Eq_U : \B \ra Rel(\E)$
  defined by $Eq_U = \Sigma_{\delta} K$. The functor $Eq_U$ maps each
  morphism $f$ to the unique morphism above $f \times f$ induced by
  the naturality of $\delta$ at $f$ and the opcartesian morphism
  $(\delta_X)_{\S}^{KX}$. If $Eq_U$ has a left adjoint $Q_U$, then $Q_U$
  is called the {\em quotient functor} for $U$.
\end{definition}
\noindent
We suppress the subscripts on $Eq_U$ and $Q_U$ when convenient. The
notion of an equality-preserving lifting of a functor is then defined
as follows:

\begin{definition}
  Let $U:\E \ra \B$ be a bifibration where $\B$ has products, suppose
  $U$ has a truth functor, and let $F:\B \ra \B$ be a functor. A
  lifting $\hat{F}$ of $F$ with respect to $Rel(U)$ is said to be
  {\em equality-preserving} if $Eq\, F \cong \hat{F}\, Eq$.
\end{definition}

Just as truth-preserving liftings are the key to defining induction
rules, equality-preserving liftings are the key to defining
coinduction rules. The following definition is pleasantly dual to
Definition~\ref{def:ind}:

\begin{definition}\label{def:coind}
  Let $U:\E\to\B$ be a bifibration where $\B$ has products, suppose
  $U$ has a truth functor, and let $F:\B\to\B$ be a functor whose
  final coalgebra has carrier $\nu F$. We say that an $Eq$-preserving
  lifting $\hat F$ of $F$ \emph{defines a sound coinduction rule for
    $\nu F$ in $U$} if the functor $Eq\dcoalg_F:\coalg_F\to\coalg_{\hat
    F}$ sending each $F$-coalgebra $\alpha : X \rightarrow F X$ to the
  $\hat F$-coalgebra $Eq\,\alpha : Eq \,X \ra Eq F X \cong \hat F Eq
  X$ preserves terminal objects.
\end{definition}

\noindent
As before, we omit explicit reference to $U$ when it is clear from context.

As in~\cite{hj98}, there is a simple condition under which
$Eq$-preserving liftings define sound coinduction rules, namely, that
$U$ has a quotient functor. Note the duality: in the inductive setting
the truth functor $K$ must have a right adjoint, whereas in the
coinductive setting the equality functor $Eq$ must have a left
adjoint.

\begin{theorem}\label{thm:coind}
  Let $U:\E \ra \B$ be a bifibration where $\B$ has products, suppose
  $U$ has a truth functor and a quotient functor, and let $F :\B \to
  \B$ be a functor whose final coalgebra has carrier $\nu F$. Then
  every equality-preserving lifting $\hat{F}$ of $F$ with respect to
  $Rel(U)$ defines a sound coinduction rule for $\nu F$.
\end{theorem}

As before, Hermida and Jacobs' proof is conceptually simple: If $U$
has a quotient functor, then each functor $Eq\dcoalg_F$ has a left
adjoint and hence preserves all limits, including the terminal
object. As a result, the carrier of the final $\hat{F}$-coalgebra is
obtained by applying $Eq$ to the final $F$-coalgebra, and the generic
fibrational coinduction rule is therefore given by
\[\mathit{coind}_F : (\forall R : Rel(\E)) (R \ra \hat{F} R) \ra R \ra
Eq (\nu F)\] 
\noindent
Soundness of the rule ensures that if $\alpha : R \ra \hat{F}R$ is
above $f$, then $\mathit{coind}_F R \alpha$ is above $\mathit{unfold}
\,f$.

As was the case for induction, Hermida and Jacobs provided
$Eq$-preserving liftings only for polynomial functors, and thus sound
coinduction rules only for carriers of their final coalgebras.  The
outstanding issue is then to establish a set of conditions under which
functors are guaranteed to have equality-preserving liftings.

\subsection{Generic Coinduction For All Coinductive Types}

The first contribution of this paper is to give a sound coinduction
rule for every coinductive type, i.e., for every data type that
is the carrier $\nu F$ of the final coalgebra for a functor $F$. This
entails determining conditions sufficient to guarantee that functors
have equality-preserving liftings.  To do this, we step back a little
and show how to construct liftings that can be instantiated to give
both the truth-preserving liftings required for deriving sound
induction rules and, by duality, the equality-preserving liftings
required for deriving sound coinduction rules.

\begin{lemma}\label{lem:coindlift}
  Define a {\em quotient category with equality} (QCE) to be a
  fibration $U:\E\to\B$ with a full and faithful functor $E:\B\to\E$
  such that $U E = \mathit{Id}_\B$ and $E$ has a left adjoint $Q$ with
  unit $\eta$. Let $F : \B \to \B$ be a functor, and define functors
  $\rho$, $J$, and $\cech{F}$ by
\[\begin{array}{lll}
  \rho:\E\to\B^\to & \;\;\;\;\;\; J:\B^\to\to\E & \;\;\;\;\;\; \check
  F : \E\to\E\\ 
  \rho P = U \eta_P &\;\;\;\;\;\; J\,(f: X \ra Y) = f^*EY &
  \;\;\;\;\;\;\check F = J\, F^\ra \, \rho  
\end{array}\]
Then $U \check{F} = F U$ (i.e., $\cech F$ is a lifting of $F$) and
$\check{F} E \cong E F$.
\end{lemma}
\begin{proof}
  To prove $U\check{F} = FU$, note that the morphisms $\rho\,P$ each
  have domain $U P$, that $\mathit{dom}\, F^\ra\, \rho = F U$, and
  that $U J = \mathit{dom}$. Together these give $U \check{F} = U J
  F^\ra \rho = F U$.  To prove $\check{F} E \cong E F$, we first
  assume that i) for every $X$ in $\B$, $\rho E X$ is an isomorphism
  in $\B$, and ii) for every isomorphism $f$ in $\B$, $J\,f \cong
  E(\mathit{dom}\,f)$. Then since $U E = \mathit{Id}_\B$, we have that
  i) and ii) imply that $\check{F} E = J F^\ra \rho E \cong E\,
  \mathit{dom}\, F^\ra \rho E = E F U E = E F$. To discharge
  assumption i), note that the counit $\epsilon:QE\to Id$ of $Q\dashv
  E$ is a natural isomorphism because $E$ is full and faithful. We
  thus have that $E\epsilon$ is also a natural isomorphism and, using
  the equality $E \epsilon \, . \, \eta E = id_E$, that $\eta E$ is a
  natural isomorphism as well. As a result, $\rho E = U\eta E$ is a
  natural isomorphism.
%{\bf Clem: Remove note
%  that, since $E$ is full and faithful, $\eta E: E \ra EQE$ is $E
%  \kappa$ for a natural transformation $\kappa:Id_\E \ra QE$, where
%  each $\kappa_X$ is an isomorphism with inverse $\epsilon_X$ and
%  $\epsilon$ is the counit of $Q\dashv E$. Then $\rho E X = U \eta_{E
%    X} = U E \kappa_X = \kappa_X$, so that $\rho EX$ is indeed an
%  isomorphism}. 
To discharge ii), let $f$ be an isomorphism in $\B$. Since cartesian
morphisms above isomorphisms are isomorphisms, we have $J f = f^*(E \,
(\mathit{cod} f)) \, \cong \, E \, (\mathit{cod} f) \, \cong \, E
(\mathit{dom} f)$.
 %f\xrightarrow[f^\S]{\cong} E\,(\mathit{cod}\,f)\, \cong
  %\xrightarrow[Ef^{-1}]{\cong} E\, (\mathit{dom}\,f)$. 
  Here, the first isomorphism is witnessed by $f^\S$ and the second by
  $Ef^{-1}$.
\end{proof}

Although it is not needed in our work, we observe that if $U$ is a QCE,
then $\rho$ is left adjoint to $J$. The proof is a straightforward
application of the universal property of reindexing; see Lemma~2.2.10
in~\cite{fum12}. The lifting $\cech{F}$ has as its dual the lifting
$\hat{F}$ given in the following lemma.

\begin{lemma}\label{lem:indlift}
  Let $U:\E\to\B$ be an opfibration, let $K:\B\to\E$ a full and
  faithful functor such that $U K = Id_\B$, and let $C:\E\to\B$ be a
  right adjoint to $K$ with counit $\epsilon$. Let $F : \B \to \B$ be
  a functor, and define functors $\pi$, $I$, and $\hat{F}$ by
\[\begin{array}{lll}
  \pi:\E\to\B^\to & \;\;\;\;\;\; I:\B^\to\to\E & \;\;\;\;\;\; \hat
  F : \E\to\E\\ 
  \pi P = U \epsilon_P &\;\;\;\;\;\; I\,(f: X \ra Y) = \Sigma_f K Y &
  \;\;\;\;\;\;\hat F = I\, F^\ra \, \pi
\end{array}\]
Then $U \hat{F} = F U$ (i.e., $\hat F$ is a lifting of $F$) and $\hat
F K \cong K F$.
\end{lemma}
\begin{proof}
  By dualisation of Lemma~\ref{lem:coindlift}. The setting on the left
  below with $U$ an opfibration is equivalent to the setting on the
  right with $U$ a fibration.
  \[\xymatrix{ \ar@{}"1,2"-<10px,0px>;"2,1"+<0px,10px>|\vdash &
    \ar@/_1.2pc/[ld]_C \E \ar[d]^{U} &
    \ar@{}"1,2"-<10px,0px>;"2,1"+<168px,10px>|\dashv &
    \ar@/_1.2pc/[ld]_C \E^{op} \ar[d]^{U}\\ \B \ar[ru]_K
    \ar[r]_{Id_\B} & \B & \B^{op} \ar[ru]_K \ar[r]_{Id_{\B^{op}}} &
    \B^{op}} \]
\end{proof}

We can instantiate Lemmas~\ref{lem:coindlift} and~\ref{lem:indlift} to
derive both the truth-preserving lifting for all functors
from~\cite{gjf10} (presented above) and an equality-preserving lifting
for all functors. The latter gives the sound induction rules for
inductive types presented in~\cite{gjf10}, and the former gives our
sound coinduction rules for all coinductive types. To obtain the
lifting for induction, let $U:\E \ra \B$ be a Lawvere category, $K$ be
the truth functor for $U$, and $C$ be the comprehension functor for
$U$. Since a Lawvere category is an opfibration,
Lemma~\ref{lem:indlift} ensures that any functor $F:\B \ra \B$ lifts
to a truth-preserving lifting $\hat{F}:\E \ra \E$. This is exactly the
lifting of~\cite{gjf10}. To obtain the lifting for coinduction, let
$U:\E \ra \B$ be a bifibration satisfying the Beck-Chevalley
condition, let $\B$ have products, and let $K$ be a truth functor for
$U$. Now, consider the relations fibration $Rel(U)$ for $U$, and let
$Eq$ be the equality functor for $U$. Since $\delta$ is a mono, since
$Eq = \Sigma_{\delta}\, K$, and since both $K$ and $\Sigma_{\delta}$
are full and faithful, Lemma~\ref{lem:bcmono} ensures that $Eq$ is
full and faithful. Moreover, since, for every $X$ in $\B$, $Eq \,X$ is
in the fibre of $Rel(U)$ above $X$, we have $Rel(U)\, Eq =
\mathit{Id}_{\B}$. We can therefore take $E$ to be $Eq$ in
Lemma~\ref{lem:coindlift} provided $Eq$ has a left adjoint $Q$.  In
this case, every functor $F:\B \ra \B$ has an equality-preserving
lifting $\check{F}:Rel(\E) \ra Rel(\E)$, and so if $F$ has a final
coalgebra $\nu F$, then $\nu F$ has a sound coinduction rule. We
record this in the following theorem. Henceforth, we call a QCE of the
form $Rel(U)$ obtained by change of base of $U$ along $\Delta$ by the
above construction, and for which the functor $E$ is thus the equality
functor for $U$, a {\em relational QCE}.

\begin{theorem}\label{thm:coinda}
If $Rel(U): Rel(\E) \to \B $ is a relational QCE obtained from a
fibration $U:\E \ra \B$, and if $F : \B \to \B$ is a functor whose
final coalgebra has carrier $\nu F$, then there exists a sound
coinduction rule for $\nu F$ in U.
\end{theorem}

Just as $\mathit{cod}$ is the canonical CCU, if $\mathit{dom}$ is the
canonical QCE. Indeed, if $U$ is $\mathit{dom} : \B^\to \to \B$, if
$E$ is the functor mapping each $X$ in $\B$ to $\mathit{id}_X$, and if
$Q$ is $\mathit{cod} : \B^\to \to \B$, then $\check F$ is exactly
$F^\ra$. Thus, just as the lifting $\hat{F}$ with respect to an
arbitrary fibration $U$ satisfying the hypotheses of
Lemma~\ref{lem:indlift} can be modularly constructed from the specific
lifting $F^\to$ with respect to $\mathit{cod}$~\cite{gjf10}, so the
lifting $\check{F}$ with respect to an arbitrary fibration $U$
satisfying the hypotheses of Lemma~\ref{lem:coindlift} can be
modularly constructed from the specific lifting $F^\to$ with respect
to $\mathit{dom}$.

What we have seen is that $\mathit{dom}$ plays a role in the
coinductive setting similar to that played by $\mathit{cod}$ in the
inductive one. We think of a morphism $f:X \ra Y$ in the total
category of $\mathit{cod}$ as a predicate on $Y$ whose proofs
constitute $X$. Intuitively, $f$ maps each $p$ in $X$ to the element
$y$ in $Y$ about which it is a proof. Similarly, we think of a
morphism $f:X \ra Y$ in the total category of $\mathit{dom}$ as a
relation on $X$, the quotient of $X$ by which has equivalence classes
comprising $Y$. Intuitively, $f$ maps each $x$ in $X$ to its
equivalence class in that quotient.

The following two examples of relational QCEs appear in
Propositions~4.8.6 and~4.8.7(iii) in~\cite{jac99}.

\begin{example}
  Let $\B$ be a regular category. The relations fibration for the
  subobject fibration $U:Sub(\B)\to\B$ is a relational QCE iff $\B$
  has coequalisers. In this case, the equality functor maps an object
  $X$ of $\B$ to the equivalence class of $\delta_X:X\to X\ti X$ in
  $Sub(\B)$. The quotient functor maps an equivalence class $[m]$,
  with $m = \langle m_0,m_1 \rangle :R\mto X\ti X$, to the codomain
  $X/R$ of the coequaliser $c_R$ of $m_0$ and $m_1$:
  \[\xymatrix{R\ar@/^/[r]^{m_0}\ar@/_/[r]_{m_1} & X
      \ar[r]^{\!\!\!c_R} & X/R}\]
  \end{example}
  
  \begin{example}
    Consider the fibration $U:\ascl\to\cl$ of admissible subsets of
    complete lattices and its associated relations fibration
    $U':Rel(\ascl)\to\cl$.  We have that $U'$ is a relational
    QCE. Indeed, the equality functor $Eq:\cl\to Rel(\ascl)$ maps a
    complete lattice $X$ to the admissible subset $\{ (x,x)\ |\ x\in
      X\} \subseteq X\ti X$.  The quotient functor
      $Q:Rel(\ascl)\to\cl$ maps an admissible subset $R \subseteq X\ti
      X$ to the complete lattice $\{x\in X\,|\,\forall (y,y')\in
      R,\ y\leq x\ \text{iff}\ y'\leq x\}$.
  \end{example}

\begin{example}\label{ex:famrel}
  As we have seen, if $U$ is the families fibration, then the fibre
  above $X$ in $Rel(U)$ consists of functions $R:X \times X \ra
  Set$. We think of these functions as constructive relations, with
  $R(x,x')$ giving the set of proofs that $x$ is related to $x'$. In
  Lemma~\ref{lem:coindlift} we can take $U$ to be the families
  fibration, $E$ to map each set $X$ to the relation $eq_X$ defined by
  $eq_X (x,x') = 1$ if $x = x'$ and $eq_X (x,x') = 0$ otherwise, and
  $Q$ to map each relation $R:X \times X \ra Set$ to the quotient
  $X/R$ of $X$ by the least equivalence relation containing $R$. We
  can then instantiate Lemma~\ref{lem:coindlift} by taking $\rho :
  Rel(U) \to Set^\to$ to map a relation $R:X \times X \ra \Set$ to the
  quotient map $\rho_R:X \ra X/R$, taking $J : Set^\to \to Rel(U)$ to
  map $f:X \ra Y$ to the relation $\bar{f}$ mapping $(x,x')$ to $1$ if
  $fx = fx'$ and to $0$ otherwise, and taking $\check{F} : FA \times
  FA \to Set$ to be given by $\check{F} R = \overline{F\rho_R}$.
\end{example}

The following example demonstrates that our approach goes beyond the
current state-of-the-art. We derive the coinduction rule for finitary
hereditary sets in the relations fibration for the families
fibration. Finitary hereditary sets are elements of the carrier of the
final coalgebra of the functor $\Pfi$ mapping a set to its finite
powerset. Since $\Pfi$ is not polynomial, it lies outside the scope of
Hermida and Jacobs' work~\cite{hj98}. In fact, as far as we aware, the
coinduction rule for finitary hereditary sets that we derive in the
next example is more general than any appearing elsewhere in the
literature; indeed, the relations in $Rel(U)$ are not required to be
equivalence relations. The functor $\Pfi$ is, however, important, not
least because a number of canonical coalgebras are built from it. For
example, a finitely branching labelled transition system with state
space $S$ and labels from an alphabet $A$ is a coalgebra with carrier
$S$ for the functor $\Pfi (A \times -)$.

\begin{example}\label{ex:coind}
  By Example~\ref{ex:famrel}, the lifting $\check{\Pfi}$ maps a
  relation $R:A \times A \ra Set$ to the relation $\check{\Pfi}R: \Pfi
  A \times \Pfi A \ra Set$ defined by $\check{\Pfi} R = \overline{\Pfi
    \rho_R}$. Thus, if $X$ and $Y$ are finite subsets of $A$, then $(X
  , Y) \in \check{\Pfi} R$ iff $\Pfi \rho_R X = \Pfi \rho_R Y$. Since
  the action of $\Pfi$ on a morphism $f$ maps any subset of the domain
  of $f$ to its image under $f$, $\Pfi \rho_R X = \Pfi \rho_R Y$ iff
  $(\forall x:X).\,(\exists y:Y).\, x\tilde{R}y \;\wedge\; (\forall
  y:Y).\, (\exists x:X).\,x\tilde{R}y$, where
  $\tilde{R}=\overline{\rho_R}$ is the least equivalence relation
  containing $R$. From $\check{\Pfi}$ we have that the resulting
  coinduction rule has as its premises a $\check{\Pfi}$-coalgebra,
  i.e., a function $\alpha: A \ra \Pfi A$ and a function $\alpha^\hash
  : (\forall a, a': A). \, aRa' \to (\alpha a) \, \check{\Pfi} R \,
  (\alpha a')$. If we regard $\alpha : A \ra \Pfi A$ as a transition
  function, i.e., if we define $a \ra b$ iff $b \in \alpha a$, then
  $\alpha^\hash$ is a bisimulation whenever $R$ is an equivalence
  relation. In this case, the coinduction rule thus asserts that any
  two bisimilar states have the same interpretation in the final
  coalgebra.  However, when $R$ is not an equivalence relation,
  $\alpha^\hash$ is slightly weaker since it only requires
  transitions to map $R$-related elements $\tilde{R}$-related
  elements.  Since it is easier to prove that two elements are $\tilde
  R$-related than it is to prove them $R$-related, our coinduction
  rule is slightly stronger than might be expected at first glance.
\end{example}

\section{Indexed Induction}\label{sec:indind}

Data types arising as initial algebras and final coalgebras on
traditional semantic categories such as $Set$ and $\omega cpo_{\bot}$
are of limited expressivity. More sophisticated data types arise as
initial algebras of functors on their indexed versions. To build
intuition about the resulting {\em inductive indexed types}, first
consider the inductive type $\mathit{List}\,X$ of lists of $X$. It is
clear that defining $\mathit{List}\, X$ for some particular type $X$
does not require any reference to $\mathit{List}\, Y$ for $Y\neq
X$. That is, each type $\mathit{List}\, X$ is inductive all on its
own. We call $\mathit{List}$ an {\em indexed inductive type} to
reflect the fact that it is a family of types, each of which is
inductive. By contrast, for each $n$ in $\mathit{Nat}$, let $\Fin\,n$
be the data type of $n$-element sets, and consider the inductive
definition of the $\Nat$-indexed type $\Lam : \Nat \ra \Set$ of
untyped $\lambda$-terms up to $\alpha$-equivalence with free variables
in $\Fin \,n$. This type is given by
\[\proofrule{i : \Fin \,n}
             {\mathit{Var} \, i : \Lam\, n }
\hspace{1cm}
\proofrule{ f : \Lam \,n &\;\; a : \Lam\, n}
          {\mathit{App} \, f \, a : \Lam\, n}
\hspace{1cm}
\proofrule{ b : \Lam \,(n+1) }
          {\mathit{Abs} \, b : \Lam\, n }
\]
Unlike $\mathit{List}\, X$, the type $\Lam \,n$ cannot be defined in
isolation using only the elements of $\Lam \,n$ that have already been
constructed. Indeed, the third inference rule above shows that
elements of $\Lam\, (n+1)$ are needed to construct elements of $\Lam
\,n$. In effect, then, all of the types $\Lam\,n$ must be inductively
constructed simultaneously. We call $\Lam$ an {\em inductive indexed
  type} to reflect the fact that it is an indexed type that is defined
inductively.

There is considerable interest in inductive and coinductive indexed
types. If types are interpreted in a category $\B$, and if $I$ is a
set of indices considered as a discrete category, then an inductive
$I$-indexed type can be modelled by the initial algebra of a functor
on the functor category $I \ra \B$. Alternatively, indices can be
modelled by objects of $\B$, and inductive $I$-indexed types can be
modelled by initial algebras of functors on slice categories
$\B/I$. Coinductive indexed types can similarly be modelled by final
coalgebras of functors on slice categories.

Initial algebra semantics for inductive indexed types has been
developed extensively~\cite{dyb94,ma09}. Pleasingly, no fundamentally
new insights were required: the standard initial algebra semantics
needed only to be instantiated to categories such as $\B/I$. By
contrast, the theory of induction for inductive indexed types has
received comparatively little attention. The second contribution of
this paper is to derive sound induction rules for inductive indexed
types by similarly instantiating the fibrational treatment of
induction to appropriate categories.  The key technical question to be
solved turns out to be the following: Given a Lawvere category of
properties fibred over types, can we construct a new Lawvere category
fibred over indexed types from which sound induction rules for
inductive indexed types can be derived? To answer this question, we
first make the simplifying assumption that the inductive indexed types
of interest arise as initial algebras of functors on slice categories,
i.e., functors $F:\B/I \ra \B/I$, where $I$ is an object of $\B$. We
treat the general case in Section~\ref{sec:fibind}.

We conjecture that the total category of the fibration with base
$\B/I$ that we seek should be a slice category of $\E$. We therefore
make the canonical choice to slice over $KI$, where $K$ is the truth
functor for $U$. We then define $U/I:\E/KI \ra \B/I$ by $(U/I) \, (f:P
\ra KI)\, =\, (Uf : UP \ra I)$. Here, $\mathit{cod} \,(Uf)$ really is
$I$ because $UK = \mathit{Id}$.

We first show that $U/I$ is indeed a bifibration. We do this by
proving a more general result that we can reuse in
Section~\ref{sec:indcoind}. 

\begin{lemma}\label{lem:slbif}
  Let $U : \E \to \B$ be a fibration (bifibration) with a functor
  $F:\A\to \E$ and $G:\A\to \B$ such that $UF = G$. This, of course,
  uniquely determines $G$. For any $I$ in $\A$, the functor
  $U/F:\E/FI\to \B/G I$ is a fibration (resp., bifibration).
\end{lemma}
\begin{proof}
  Let $\alpha:Y \ra GI$ and $\beta:X \ra GI$ be objects of $\B/GI$, and
  let $\phi:Y \ra X$ be a morphism in $\B/GI$ from $\alpha$ to $\beta$,
  i.e., let $\phi$ be such that $\alpha = \beta \phi$. Let $f:P \ra
  FI$ be an object of $\E/FI$ such that $(U/F) f = U f = \beta$, and
  let $\phi^\S_P:\phi^*P \ra P$ be the cartesian morphism in $\E$ above
  $\phi$ with respect to $U$. Then $\phi^\S_P$ is a morphism in
  $\E/FI$ with domain $f \phi^\S_P$ and codomain $f$, and it is
  cartesian above $\phi$ with respect to $U/F$. Thus, $U/F$ is a
  fibration if $U$ is. Now, let $g:Q \ra FI$ be an object of $\E/FI$
  such that $(U/F) g = U g = \alpha$, and let $\phi^Q_\S: Q \ra
  \Sigma_\phi Q$ be the opcartesian morphism in $\E$ above $\phi$ with
  respect to $U$. Since $\alpha = \beta \phi$, the opcartesianness of
  $\phi^Q_\S$ ensures that there is a unique morphism $k:\Sigma_\phi\, Q
  \ra FI$ in $\E$ above $\beta$ such that $g = k \phi^Q_\S$. Then
  $\phi^Q_\S$ is a morphism in $\E/FI$ with domain $g$ and codomain
  $k$, and it is opcartesian above $\phi$ with respect to $U/F$. Thus,
  $U/F$ is an opfibration if $U$ is. Combining these results gives
  that if $U$ is a bifibration then so is $U/F$.
\end{proof}

We can now show that $U/I$ is a bifibration as desired.

\begin{lemma}\label{lem:bif}
  If $U : \E \to \B$ is a fibration (bifibration) with a truth functor $K$
  and $I$ is an object of $\B$, then $U/I$ is a fibration (resp.,
  bifibration).
\end{lemma}
\begin{proof}
This follows from Lemma~\ref{lem:slbif} by taking $F$ to be the truth
functor $K$ for $U$ and $G$ to be $\mathit{Id}_\B$, and then observing
that, for this instantiation, $U/F$ is precisely the fibration $U/I$
defined before Lemma~\ref{lem:slbif}.
\end{proof}

There is an alternative characterisation of $\E/KI$ that both
clarifies the conceptual basis of our treatment of indexed induction
and simplifies our calculations. The next lemma is the key observation
underlying this characterisation.
\begin{lemma}\label{lem:simindex}
  Let $U : \E \to \B$ be a fibration with truth functor $K$, let $I$
  be an object of $\B$, and let $\alpha:X \ra I$. Then
  $(\E/KI)_{\alpha} \cong \E_X$.
\end{lemma}
\begin{proof}
  One half of the isomorphism maps the object $f:P \ra KI$ of
  $(\E/KI)_{\alpha}$ to $P$. For the other half, note that since truth
  functors map objects to terminal objects, and since reindexing
  preserves terminal objects, we have that $\alpha^* KI$ is terminal
  in $\E_X$. Thus, for any object $Q$ above $X$, we get a morphism
  from $Q$ to $KI$ by composing $\alpha^\S_{KI}$ and the unique
  morphism $!$ from $Q$ to $\alpha^* KI$. Since $!$ is vertical and
  $\alpha^\S_{KI}$ is above $\alpha$, this composition is above
  $\alpha$. Thus each object $Q$ in $\E_X$ maps to an object of
  $(\E/KI)_{\alpha}$. It is routine to verify that these maps
  constitute an isomorphism.
\end{proof}
By Lemma~\ref{lem:simindex} we can identify objects (morphisms) of
$(\E/KI)_{\alpha}$ and objects (resp., morphisms) of $\E_X$.  This
gives our abstract characterisation of $U/I$:
\begin{lemma}\label{lem:U/Icob}
  Let $U : \E \to \B$ be a fibration with a truth functor and let $I$
  be an object of $\B$.  Then $U/I$ can be obtained by change of base
  of $U$ along $\mathit{dom}:\B/I \ra \B$.
\end{lemma}
\begin{proof}
  As noted in Section~\ref{sec:coind}, the pullback of a fibration
  along a functor is a fibration. The objects (morphisms) of the fibre
  above $\alpha:X \ra I$ of the pullback of $U$ along $\mathit{dom}$
  are the objects (resp., morphisms) of $\E_X$. By
  Lemma~\ref{lem:simindex}, the pullback of $U$ along $\mathit{dom}$
  is therefore $U/I$.
\end{proof}

As observed just after Definition~\ref{def:cob}, pulling back a
fibration along a functor preserves fibred terminal objects so, by
Lemma~\ref{lem:U/Icob}, $U/I$ has fibred terminal objects if $U$
does. Concretely, the truth functor $K_{U/I}:B/I \ra \E/KI$ maps an
object $f:X \ra I$ to $Kf:KX \ra KI$. To see that $U/I$ is a Lawvere
category if $U$ is, we must also show that $K_{U/I}$ has a right
adjoint if $K_U$ does. For this, we use an abstract theorem due to
Hermida~\cite{her93} to transport adjunctions across pullbacks along
fibrations.

\begin{lemma}\label{lem:claudio}
  Let $F \dashv G:\A \ra \B$ be an adjunction with counit $\epsilon$,
  and let $U:\E \ra \B$ be a fibration. Then the functor $U^*F:U^*\A
  \ra \E$ has a right adjoint $G_U:\E \ra U^*\A$ mapping each object
  $E$ to the object $(\epsilon_{UE}^* E, G U E)$.
\end{lemma}

\begin{lemma}\label{lem:pb-pres}
  Change of base along a fibration preserves CCUs, i.e., if $U:\E \ra
  \B$ is a CCU and $U': \E' \ra \B$ is a fibration, then the pullback
  $U'^*U$ is a CCU.
\end{lemma}
\begin{proof}
  We already have that $U'^*U$ is a fibration with fibred terminal
  objects. To see that $K_{U'^*U}$ has a right adjoint, consider the
  pullback of $K_U$ along $U^*U'$.  This pullback is given by $\E'$,
  $K_{U'^*U} : \E' \ra U'^*\E$, and $U' : \E' \ra \B$. Note that
  $U^*U'$ is a fibration since it is obtained by pulling $U'$ back
  along $U$. Lemma~\ref{lem:claudio} then ensures that, since $K_U$
  has a right adjoint, so does $K_{U'^*U}$. Thus $U'^*U$ is a CCU.
\end{proof}

If $U : \E \to \B$ is a fibration, $I$ is an object of $\B$, and $U'$
is $\mathit{dom} : \B/I \to \B$, then the comprehension functor for
$U'^*U$ --- i.e., for $U/I$ --- maps an object $f:P \ra KI$ to $(Uf)
\pi_P : \{P\} \to I$. Combining Lemma~\ref{lem:pb-pres} and the fact
that change of base preserves bifibrations, we have:

\begin{lemma}\label{lem:coblwfib}
  Let $U:\E \ra \B$ be a Lawvere category and $U':\E' \ra \B$ be a
  fibration. Then $U'^*U$ --- i.e., $U/I$ --- is a Lawvere category.
\end{lemma}
\noindent

\begin{theorem}\label{thm:iinda}
Let $U : \E \to \B$ be a Lawvere category, let $I$ be an object of
$\B$, and let $F : \B/I \to \B/I$ be a functor whose initial algebra
has carrier $\mu F$. Then there exists a sound induction rule for $\mu
F$ in $U/I$.
\end{theorem}

We can use Theorem~\ref{thm:iinda} to derive a sound induction rule
for the indexed containers of Morris and Altenkirch~\cite{ma09}.

\begin{example}\label{ex:indind}
  If $I$ is a set, then the {\em category of $I$-indexed sets} is the
  fibre $\Fam(\Set)_I$. An {\em $I$-indexed set} is thus a function
  $X:I \ra \Set$, and a morphism $h$ from $X$ to $X'$, written $h : X
  \ra_I X'$, is a function of type $(\Pi i\!:\!I).\, X i \ra
  X'i$. Morris and Altenkirch denote this category $I \ra \Set$ and
  define an {\em $I$-indexed container} to be a pair $(S,P)$ with $S :
  I \ra \Set$ and $P : (\Pi i\!:\!I).\, Si \ra I \ra \Set$. An
  $I$-indexed container defines a functor $[S,P]: (I \ra \Set) \ra I
  \ra \Set$ by $[S,P] X i = (\Sigma s\!:\!Si).\, P \, i \, s \ra_{I}
  X$.  Thus, if $t \!:\! [S,P]\, X \, i$, then $t$ is of the form
  $(s,f)$. If $g : X \ra_I Y$ is a morphism of $I$-indexed sets, then
  $[S,P] g$ maps a pair $(s,f)$ to $(s,gf)$.

  If $i \in I$, then we can think of $S i$ as a collection of
  operators that produce data of sort $i$, and we can think of $P$ as
  assigning to every $i$ and every operator producing data of sort $i$
  an $I$-indexed collection of positions in which data is stored. That
  is, $P \ i \, s \, j$ is the set of positions associated with the
  operator $s$ where data of sort $j$ must be stored. This {\em shapes
    and positions metaphor} is also reflected in the functor
  associated with an indexed container, since we can think of $[S,P]X
  i$ as containing terms of sort $i$ produced by $(S,P)$ whose input
  data is drawn from $X$. Such a term consists of an operator $s$
  producing data of sort $i$ and, for each position storing data of
  sort $j$, an element of $X$ of sort $j$.

  The initial algebra of $[S,P]$ is denoted $\mathit{in}:[S,P] W_{S,P}
  \ra_I W_{S,P}$.  Since $I \ra \Set$ is equivalent to $\Set/I$, we
  can use the results of this section to extend those of~\cite{ma09}
  by giving sound induction rules for data types of the form
  $W_{S,P}$. A predicate over an $I$-indexed set $X$ is a function $Q:
  (\Pi i:I).\, X i \ra \Set$.  To simplify notation, this is written
  $Q:X \ra_I \Set$. The lifting $\widehat{[S,P]}$ of $[S,P]$ maps each
  $Q:X \ra_I \Set$ to the predicate $\widehat{[S,P]}Q : [S,P]X \ra_I
  \Set$ defined by $ \widehat{[S,P]} \, Q \, i \, (s,f) = (\Pi
  j\!:\!I).\, (\Pi p\!:\! P \, i \, s \, j). \, Q \, j \, (f \, j \,
  p)$.  Altogether, this gives the following sound induction rule for
  $W_{S,P}$:
\[\begin{array}{l}
  (\Pi i\!:\!I).\,(\Pi (s,f)\!:\![S,P]\, W_{S,P}\, i).\, 
  ((\Pi j\!:\!I).\,(\Pi p\!:\! P \, i \, s \, j). \, Q \, j (f \, j \,
  p) \ra Q \, i (\mathit{in} \, i \, (s,f)))) \\ 
  \;\;\;\;\;\;\;\;\;\;\;\;\;\ra 
  (\Pi i\!:\!I).\,(\Pi t\!:\! W_{S,P}\; i).\, Q \, i \, t
\end{array}
\]

\vspace*{0.15in}

While admittedly rather dense in its type-theoretic formulation, the
above induction rule is conceptually clear. The premise says that, for
any term $\mathit{in}\,i\,(s,f)$ in $W_{S,P}i$, we must be able to
prove that a property $Q:W_{S,P} \ra_I \Set$ holds at
$\mathit{in}\,i\,(s,f)$ if $Q$ is assumed to hold of all the immediate
subterms of $\mathit{in}\,i\,(s,f)$. The conclusion of the rule says
that $Q$ holds for all terms. Of course this is what we naturally
expect, and our point is precisely that we can {\em derive it in a
  principled manner} from the fibrational approach to induction rather
than simply having to postulate that it is reasonable.

We can instantiate the above induction rule for $W_{S,P}$ for the data
type of untyped lambda terms from the beginning of this section. The
resulting induction rule cannot be derived using Hermida and Jacobs'
techniques because the data type of untyped lambda terms is not the
initial algebra of a polynomial functor. The resulting rule is
precisely what we expect. For any predicate $Q:\Lam \ra_{\Nat} \Set$:

\[\begin{array}{l}
(\Pi n :\Nat. \; \Pi j : \Fin \; n. \; Q \; n \; (Var
\; j)) \ra \\ 
(\Pi n :\Nat. \; \Pi u, v : \Lam \; n.  \; Q \; n \; u
\ra  Q \; n \; v
\rightarrow Q \; n \; (App \; u \; v)) \rightarrow\\
(\Pi n :\Nat. \; \Pi t :  \Lam \; (n\! +\! 1). \; Q \; (n\! + \! 1) \; t
\rightarrow Q \; n \; (Abs \; t)) \ra \\
\Pi n :\Nat. \; \Pi t :  \Lam \; n.  \; Q \; n \; t \\
\end{array}
\]
\end{example}

\section{Indexed Coinduction}\label{sec:indcoind}

We now present our third contribution, namely sound coinduction rules
for coinductive indexed types. Examples of such types are infinitary
versions of inductive indexed types, such as infinitary untyped lambda
terms and interaction structures. Following the approach of
Section~\ref{sec:indind}, we consider indexing by slice categories in
this section. In more detail, we show that for any relational QCE over
a base category $\B$ and for any object $I$ of $\B$, change of base
along $dom:\B/I \ra \B$ yields a relational QCE over $\B/I$.

Recall that if $\B$ has products and $U:\E \to \B$ is a bifibration
that satisfies the Beck-Chevalley condition and has truth functor $K$,
then the equality functor $Eq$ for $U$ is given by $Eq =
\Sigma_{\delta} K$. Let $Rel(U):Rel(E) \ra \B$ be a relational QCE, so
that $Eq$ has a left adjoint $Q$. To define a relational QCE over
$\B/I$ we must first see that $\B/I$ has products. But the product of
$f$ and $g$ in $\B/I$ is determined by their pullback: if $W$, $j: W
\to Z$, and $i : W \to X$ give the pullback of $f$ and $g$, then their
product in $\B/I$ is the morphism $fi$ or, equivalently, $gj$. Below,
we write $f^2$ for the product of $f : X \to I$ with itself in $\B/I$
and $X_fX$ for the domain of $f^2$. Now, if $\B$ has pullbacks, then
we can construct the relations fibration $Rel(U/I):Rel(\E/KI)\to \B/I$
from the pullback of $U/I$ along the product functor $\Delta/I :\B/I
\ra \B/I$ mapping $f$ to $f^2$. Concretely, an object of $Rel(\E/KI)$
above $f:X \ra I$ is an object of $\E/KI$ above $f^2$ with respect to
$U/I$. This is, in turn, equivalent to an object of $\E$ above $X_fX$
with respect to $U$.

\subsection{The Equality Functor for $U/I$}\label{sec:eqU/I}

In Section~\ref{sec:indind} we showed that if $U$ is a bifibration
where $\B$ has products, and $U$ has a truth functor $K$, then for any
object $I$ of $\B$, $U/I$ is a bifibration that has a truth functor
whose action is also that of $K$, and so is denoted $K$ as
well. Furthermore, we have just seen that if $\B$ has pullbacks, then
$\B/I$ also has products.  Thus, by Definition~\ref{def:eq-fun}, $U/I$
has an equality functor $Eq_{U/I}$. To define this functor concretely,
note that the component of the diagonal natural transformation
$\delta/I:\mathit{Id}_{\B/I} \ra \Delta/I$ at $f:X \ra I$ is the
mediating morphism in the diagram below on the left.  Thus, $Eq_{U/I}$
maps an object $f : X \ra I$ of $\B/I$ to the unique morphism above
$f^2$ in the diagram on the right induced by the opcartesian morphism
$m$ above $(\delta/I)_f$:
\[\xymatrix{
  X \ar@/_1pc/[rdd]_{id} \ar@{.>}[rd]|{(\delta/I)_f}
  \ar@/^1pc/[rrd]^{id}\\ &X_fX \ar[d]_i \ar[r]^j \lpbc & \ar[d]^f X &
  & KI \\ &X \ar[r]_f &I & KX \ar[ru]^{Kf} \ar[r]_{m} &
  \Sigma_{(\delta/I)_f}KX \ar@{.>}[u]_{Eq_{U/I}f}}
\] 
Note that if $U$ satisfies the Beck-Chevalley condition, so that
opreindexing for $U$ defines a full and faithful functor, then the
fact that the action of opreindexing for $U/I$ is the same as
opreindexing for $U$ means that opreindexing for $U/I$ defines a full
and faithful functor as well.
Since truth functors are always full and
  faithful, so is $Eq_{U/I} = \Sigma_{\delta/I} K_{U/I}$.

\subsection{The Quotient Functor for $U/I$}

Whereas defining the equality functor for $U/I$ was straightforward,
defining its quotient functor is actually tricky.  To do so, for each
object $I$ of $\B$, we we write $Rel(U)/I : Rel(\E)/Eq\, I \ra \B/I$
for the fibration obtained as the instantiation of
Lemma~\ref{lem:slbif} in which $Eq : \B \to Rel(\E)$ plays the role of
$F$ and $Rel(U)$ plays the role of $U$. Concretely, the objects of
$Rel(\E)/Eq \, I$ above $f:X \ra I$ are morphisms $\alpha:P \ra Eq\,
I$, for some object $P$ of $Rel(\E)$, such that $U \alpha = \Delta
f$. Our first result identifies conditions under which $Rel(U)/I$ is a
QCE.

\begin{lemma}\label{lem:lift-qce}
  Let $U : \E \to \B$ be a fibration, let $\B$ have pullbacks, let $I$
  be an object of $\B$, and let $Rel(U):Rel(\E) \ra \B$ be a
  relational QCE. Then $Rel(U)/I$ is a QCE.
\end{lemma}
\begin{proof}
  Let $Eq :\B \to Rel(\E)$ and $Q : Rel(\E) \to \B$ be the equality
  and quotient functors for $U$, respectively. We construct a full and
  faithful functor $E':\B/I \ra Rel(\E)/Eq \,I$ such that $(Rel(U)/I)
  \, E' = \mathit{Id}_{\B/I}$, and a left adjoint $Q'$ for $E'$, as
  follows. Take $E'$ to be $Eq$. Then $E'$ is full and faithful since
  $Eq$ is. Moreover, for any $f : X \to I$,
  Definition~\ref{def:eq-fun} ensures that $Eq\,f$ is above $f \times
  f$ with respect to $U$, so $(Rel(U)/I)\, E' f = f$, and thus
  $(Rel(U)/I) \, E' = Id_{\B/I}$. Finally, we define $Q'$ to map each
  object $\alpha : P \ra Eq \,I$ of $Rel(\E)/Eq\, I$ to its transpose
  $\alpha': QP \ra I$ under the adjunction $Q \dashv Eq$. That $Q'
  \dashv E'$ follows directly from $Q \dashv Eq$.
\end{proof}

We can now define the quotient functor for $Rel(U/I)$ using the
functor $Q'$ from the proof of Lemma~\ref{lem:lift-qce}. The key step
is to use Lemma~\ref{lem:natadj} to define an adjunction $\tau \dashv
\sigma$ such that the following diagram commutes:
\[\xymatrix{{Rel(\E/KI)}\ar[dr]_{Rel(U/I)} \ar@/^/[rr]^\tau
  \ar@{}[rr]|\bot & & 
   \ar@/^/[ll]^\sigma \ar[ld]^{Rel(U)/I} {Rel(\E)/Eq\, I}\\ & {\B/I} &
}\] 

\vspace*{0.1in}

\noindent
Then if $E'$ and $Q'$ are the functors witnessing the fact that
$Rel(U)/I$ is a QCE, compositionality of adjoints ensures that
$\sigma E'$ and $Q'\tau$ give equality and quotient functors for
$Rel(U/I)$, respectively.

\vspace*{0.1in}

\begin{lemma}\label{lem:tausigma}
  The above adjunction $\tau \dashv \sigma$ holds. 
\end{lemma}
\begin{proof}
  In order to prove this lemma, we first instantiate
  Lemma~\ref{lem:slbif}, with $Eq$ playing the role of $F$ and
  $\Delta$ playing the role of $G$, to obtain the fibration
  $U/Eq:\E/Eq\,I\to \B/I\ti I$. We then have the following three
  changes of base:
  \[\xymatrix{ Rel(\E/KI) \ar[d]_{Rel(U/I)}\ar[r]\pbc & \E/KI
    \ar[d]_{U/I} \ar[r] \pbc & \E/Eq\,I \ar[d]|{U/Eq} & \ar[l]
    Rel(\E)/Eq\,I \ar[d]^{Rel(U)/I} \pbc[dl]\\ B/I \ar[r]_{\Delta/I} &
    B/I\ar[r]_{\delta_I\circ\_}&B/I\ti I & \ar[l]^{\_\times\_} B/I}\]
  Here, the functor $\delta_I\circ \_$ maps $f:X\to I$ to
  $\delta_I\circ f:X\to I\ti I$, and $\_\ti\_$ maps $f$ to $f\ti
  f:X\ti X\to I\ti I$. The square on the left is a pullback square by
  definition of $Rel(U/I)$, and the one on the right is a pullback
  square by direct calculation. To see that the middle square is a
  pullback square, first observe that since every morphism
  $\delta_I:I\to I\ti I$ is a mono, Lemma~\ref{lem:bcmono} ensures
  that each opcartesian morphism $(\delta_I)_\S:KI\to Eq\,I$ is also
  cartesian.  For any $f:X\to I$, the fibre above $f$ of the pullback
  of $U/Eq$ along $(\delta_I)_\S:KI\to Eq\,I$ consists of all
  morphisms of the form $g : P \to Eq I$ such that $Ug = \delta_I
  \circ f$. Similarly, the fibre of $U/I$ above $f$ consists of all
  morphisms of the form $g : P \to KI$ such that $Ug = f$. The
  universal property of $(\delta_I)_\S$ considered as a cartesian
  morphism ensures that these two fibres are isomorphic, and thus that
  $U/I$ is indeed the pullback of $U/Eq$ along $\delta_I \circ \_$.

  Now, let $f:X \ra I$ be a morphism in $\B$, and let $i$ and $j$ be
  the projections for the pullback square defining $X_f X$. The
  universal property of the product $X \times X$ ensures the existence
  of a morphism $v_f : X_fX \to X \times X$ such that $\pi_1 v = i$
  and $\pi_2 v = j$.  Moreover, by the universal property of the
  pullback of $f$ along itself, $v_f$ is a mono. In fact, it is easy
  to check that there is a natural transformation $v:dom\circ
  \Delta/I\to \Delta\circ dom$ whose component at any $f$ is given by
  $v_f$.  Finally, $v$ extends to a natural transformation $\alpha :
  (\delta_I\circ\_)\circ \Delta/I \to \_\times \_$. Indeed, for any
  $f$, the fact that $\pi_n\circ f\ti f\circ v_f = f^2 = \pi_n\circ
  \delta_I\circ f^2$ for $n\in\{1,2\}$ ensures that the diagram
  \[\xymatrix{
    X_fX \ar[d]_{f^2} \ar[r]^{v_f} & X\times X \ar[d]^{f\times f}\\
    I \ar[r]_{\delta_I} & I\times I}
  \]
  commutes by the universal property of the product $I\times I$.  By
  Lemma~\ref{lem:natadj}, $\alpha$ induces the desired adjunction.
\end{proof}

Recall that our candidate for the quotient functor $Q_{U/I}$ for
$Rel(U/I)$ is $Q'\tau$. To see that $Q'\tau \dashv Eq_{U/I}$, we need
only verify that $Eq_{U/I}$ is $\sigma E'$. It is routine to check
that $\tau Eq_{U/I} = E'$, from which $Eq_{U/I} = \sigma E'$
follows. We therefore have that $Rel(U/I)$, together with 
$Eq_{U/I}$ and $Q_{U/I}$ as defined above, form a relational
QCE. Thus, by Theorem~\ref{thm:coinda}, we have 

\begin{theorem}\label{thm:icoinda}
  Let $U :\E \to \B$, where $\B$ has products and pullbacks, be a
  bifibration that satisfies the Beck-Chevalley condition. Suppose $U$
  has a truth functor. Let $I$ be an object of $\B$ and $F : \B/I \to
  \B/I$ be a functor whose final coalgebra has carrier $\nu F$. Then
  there exists a sound coinduction rule for $\nu F$ in $U/I$.
\end{theorem}

We can use the results of this section to give a sound coinduction rule
for final coalgebras of indexed containers that is dual to the
sound induction rule of Example~\ref{ex:indind}.

\begin{example}\label{ex:coindcont}
  Let $(S,P)$ be an $I$-indexed container with final coalgebra
  $\mathit{out}: M_{S,P} \ra_I [S,P]M_{S,P}$. A relation above an
  $I$-indexed set $X:I \ra \Set$ is an $I$-indexed family of relations
  $R i$ on $X i$. The relational lifting of $[S,P]$ maps a relation
  $R$ above an $I$-indexed set $X$ to the relation $R'$ above the
  $I$-indexed set $[S,P]X$ that relates $(s,f) \in [S,P]X i$ and $(s',
  f') \in [S,P]X i$ iff $s = s'$ and, for all $j:I$ and $p\!:\!P \, i
  \, s \, j$, $f \, j \, p$ is related to $f' \, j \, p$ in the least
  equivalence relation containing $R j$. This gives the following
  notion of bisimulation for $[S,P]$-coalgebras $k:X \ra_I [S,P]X$.
  Let $\rho_0 (s,f) = s$ and $\rho_1 (s,f) = f$. Then if $x, x' \in X
  i$, then $x \sim_i x'$ iff $\rho_0 (k x) = \rho_0 (k x')$ and, for
  all $j:I$ and $p\!:\!P\, i \, (\pi_0 (k x)) \, j$, we have that
  $\rho_1 (k x) j p \sim_j \rho_1 (k x') j p$.  As in
  Example~\ref{ex:coind}, the coinduction rule thus asserts that any
  two bisimilar states have the same interpretation in the final
  coalgebra.
\end{example}

Stepping back, we see that the above coinduction rule is as
expected. To understand it, we think of a term $t : M_{S,P}$ as being
part of a transition system whose terms are the subterms of $t$
(including $t$ itself), and suppose there is a transition from every
subterm to each of the immediate subterms of that term. Then two terms
are bisimilar iff they share the same root operator and each of their
subterms are bisimilar. The point is, of course, that the fibrational
approach to coinduction derives the rule in a principled manner rather
than simply having to postulate it. 

\section{Fibred Induction}\label{sec:fibind}

In Section~\ref{sec:indind} we saw how the fibrational approach to
induction can be instantiated to derive sound induction rules for
inductive indexed types when the indexing is given using slice
categories. Although it provides a good example of how to exploit the
abstract power of fibrations, this instantiation suffers from two
limitations:

\vspace*{0.1in}

\begin{iteMize}{$\bullet$}
\item First, the instantiation successfully treats indexing that is
  modelled by slice categories, but neither it nor its abstract
  generalisation can successfully handle more general forms of
  indexing. Indeed, in Section~\ref{sec:indind}, the fact that the
  comprehension functor $\mathit{dom}$ is a fibration was critical to
  showing that $U/I$ is a Lawvere category, but the abstract
  generalisation of this result does not hold because, in general, a
  comprehension functor need not be a fibration. To handle general
  forms of indexing, we therefore need a genuinely new idea.

\item Secondly, in Section~\ref{sec:indind} we handle $I$-indexed
  types by deriving from a Lawvere category $U:\E \ra \B$ a Lawvere
  category $U/I$ with base category $\B/I$. But this is inelegant
  because it requires the construction of a new Lawvere category for
  every possible index $I$, and because the uniformity over $I$ that
  connects the different fibrations $U/I$ is completely ignored.
  Indeed, if we think of fibrations as modelling logics over types,
  then the results of Section~\ref{sec:indind} ostensibly choose
  entirely different logics for different indices. A better approach
  would formalise the uniformity of the fibrations $U/I$ over the
  indices $I$.
\end{iteMize}

\vspace*{0.1in}

\noindent In this section we extend the work of Section~\ref{sec:indind} to
derive sound induction rules for general indexed types. This is
accomplished by adding an abstraction layer that models the way types
are indexed, and thus allows us to treat indexing modelled by
structures other than slice categories. More specifically, we consider
indexed types to be given by a second fibration $r:\B\to\A$, so that
the objects of $\B$ are types indexed by the objects of $\A$. Of
course, the logical layer still forms a fibration $U:\E\to\B$ over
types, so we get the following basic picture, which captures the move
from a single fibration $U$ to a fibration $U$ above a fibration $r$:
\[\xymatrix{ \E\ar[dr]_{rU} \ar[rr]^{U}
    & & \ar[ld]^r \B\\ & \A & }\] 

\vspace{0.1in}

\noindent
Note that $rU$ is a fibration because the composition of two
fibrations is again a fibration. Taking $\mathcal{A}$ to be the
category with one object and one morphism, the sound induction rules
in the unindexed setting will be recoverable from the sound induction
rules for general indexed types that we develop in this section; see
Example~\ref{ex:sec2} below. The sound induction rules in the indexed
setting of Section~\ref{sec:indind} will similarly be recoverable by
taking $r$ to be the codomain fibration; see
Lemma~\ref{lem:canlwfiba}.  In addition, taking $r$ to be the families
fibration, we will be able to derive the sound induction rules for
carriers of initial algebras of indexed containers directly, rather
than deriving them indirectly, as we did in Example~\ref{ex:indind},
using the equivalence of the families fibration and the codomain
fibration.

Let $U : \E \ra \B$ and $r : \B \ra \A$ be two fibrations. Then an
{\em inductive indexed type} with index $a$, where $a$ is an object of
$\A$, is the carrier $\mu F$ of the initial algebra of an endofunctor
$F:\B_a \ra \B_a$, where $\B_a$ is the fibre of $\B$ above $a$. To
derive a sound induction rule for $\mu F$ using
Theorem~\ref{thm:inda}, we will ultimately need a Lawvere category
with base $\B_a$; indeed, by the discussion immediately following
Lemma~\ref{lem:deflifting}, this will ensure the existence of a
lifting of $F$ to the total category of that Lawvere category. So,
what might we take as that Lawvere category? Since $F$ has domain
$\B_a$ rather than all of $\B$, we cannot expect $F$ to lift to the
whole of $\E$. On the other hand, $U$ does restrict to a fibration
$U_a:\E_a \ra \B_a$, where $\E_a$ is the fibre of $rU$ above $a$.  As
we will see in Corollary~\ref{cor:iind} below, $U_a$ is precisely the
Lawvere category we seek.

We begin by establishing the properties of $U_a$ that we will
need. The following lemma uses change of base to deduce several of
them.

\begin{lemma}\label{lem:fibredfib}
  Let $U:\E\to\B$ and $r:\B\to\A$ be two fibrations. For any object
  $a$ in $\A$, the fibration $U$ restricts to a fibration
  $U_a:\E_a\to\B_a$, where $\E_a$ is the fibre above $a$ of the
  fibration $r U$. Similarly, if $U$ is an opfibration or a
  bifibration, then so is $U_a$. Finally, if $U$ has a truth functor,
  then so does $U_a$.
\end{lemma}
\begin{proof}
  The fibration $U_a: \E_a \ra \B_a$ arises by change of base of $U$
  along the inclusion functor $i_a:\B_a \ra \B$:
 \[\xymatrix{
   \E_a\ar[d]_{U_a}\ar[r] \lpbc & \E\ar[d]^U\\
   \B_a\ar[r]_{i_a} & \B}
 \]
Pulling back an opfibration along a functor produces an opfibration,
so $U_a$ is an opfibration if $U$ is. As a result, $U_a$ is a
bifibration if $U$ is. Moreover, change of base preserves truth
functors, so $U_a$ has a truth functor if $U$ does. Indeed, the truth
functor for $U_a$ is just the restriction of the truth functor for $U$
to $\B_a$.
\end{proof}
\noindent
We write $K_a$ for the truth functor for $U_a$. Note that while a
truth functor always restricts to a subfibration $U_a$, the existence
of a truth functor $K_a$ for every $U_a$ does not necessarily imply
that $U$ itself has a truth functor. For this to be the case,
reindexing must preserve truth functors from one subfibration to
another. Of course, if $U$ is a bifibration, then reindexing is a
right adjoint, so it preserves terminal objects, and in this case the
individual truth functors $K_a$ actually do collectively define a
truth functor for $U$.

Our interest in the above results is that they show that the basic
structure of a logic (reindexing, opreindexing, and truth functors)
over a fibration of indexed types restricts to a corresponding logic
over types with a specific index. We may therefore consider
truth-preserving (i.e., $K_a$-preserving) liftings of functors $F:\B_a
\ra \B_a$, and ask when such a lifting defines a sound induction rule
for $\mu F$.  From Theorem~\ref{thm:ind} we know the answer: this
occurs when the fibration $U_a:\E_a \ra \B_a$ is a CCU. But now we
face a choice. Is it enough to simply ask that, for every object $a$
of $\A$, $U_a$ is a CCU? Or should we require that these different
CCUs, when taken collectively, ensure that $U$ is a CCU?

While the former choice is indeed possible, we believe that the latter
choice better highlights the uniformity connecting the different
fibrations $U_a$. In fact, we have already implicitly made the latter
choice when we started with a single fibration $U$ and constructed
from it the collection of individual fibrations $U_a$. Unfortunately,
asking that each fibration $U_a$ is a CCU does not ensure that $U$
itself is a CCU.
% this is in contrast to the situation for truth functors where, when
%$U$ is a bifibration, the individual truth functors $K_a$ for the
%fibrations $U_a$ do collectively define a truth functor $K$ for $U$.
On the other hand, we cannot simply require $U$ to be a CCU either,
since that is not enough to guarantee that each $U_a$ is a CCU. But if
we require $U$ to be a fibred CCU in the sense of Definition~4.4.5 of
\cite{jac91}, then $U$ will indeed be a CCU whose restriction to each
subfibre $U_a$ is also a CCU. We have:

\begin{definition}\label{def:ccua}
  Let $U:\E\to\B$ and $r:\B\to\A$ be two fibrations, and let
  $K:\B\to\E$ the truth functor for $U$. We say that $U$ is a
  \emph{fibred CCU above $r$} if $K$ has a fibred right adjoint
  $\{-\}:r U\to r$: 
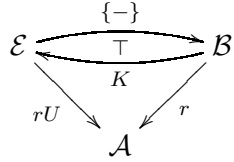
\[\xymatrix{ \E\ar[dr]_{r U} \ar@/^/[rr]^{\{-\}}
    \ar@{}[rr]|\top & & \ar@/^/[ll]^K \ar[ld]^r \B\\ & \A & }\]
\end{definition}
\noindent
That $K:r \ra rU$ is a fibred functor follows from the fact that $K:
id_\B \ra U$ is also a fibred functor (see Lemma~1.8.8
of~\cite{jac99}). A first consequence of Definition~\ref{def:ccua} is
that, if a fibration $U:\E\to\B$ is a fibred CCU above $r$ then $U$ is
a CCU. Furthermore, from Lemma~\ref{lem:fibadj} we have that if $U$
is a fibred CCU above $r$, then each fibration $U_a$ is a CCU. In
fact, we have the following correspondence: 

\begin{lemma}\label{lem:ccuafibadj}
  Let $U:\E\to\B$ and $r:\B\to\A$ be fibrations. The fibration $U$ is
  a fibred CCU above $r$ with fibred adjunction $K \dashv \{-\}$ iff
  $U$ is a CCU with truth functor $K:\B\to\E$ and comprehension
  functor $\{-\}:\E\to\B$ and, for each $a$ in $\A$, the fibration
  $U_a:\E_a\to\B_a$ is a CCU with comprehension functor
  $\{-\}_a:\E_a\to\B_a$ given by restricting $\{-\}$ to $\E_a$.
\end{lemma}

\begin{proof}
  Let $U$ be a CCU with truth functor $K:\B\to\E$ and comprehension
  functor $\{-\}:\E\to\B$.  Further, suppose that, for every $a$ in
  $\A$, the fibration $U_a:\E_a\to\B_a$ is a CCU whose comprehension
  functor $\{-\}_a:\E_a\to\B_a$ is given by restricting $\{-\}$ to
  $\E_a$.  Then, by Lemma~\ref{lem:adjandcart}, we have that $\{-\}$
  is fibred from $r U$ to $r$. Moreover, since the adjunction
  $\{-\}\vdash K$ restricts to the adjunctions $\{-\}_a\vdash K_a$,
  the unit of $\{-\}\vdash K$ is vertical with respect to $r$.  The
  other direction of the equivalence is straightforward.
\end{proof}

\noindent
We have thus shown that a fibred CCU $U$ above $r$ is just the right
structure for deriving sound induction rules when indexing of types is
described by $r$. We wanted a structure to guarantee that each $U_a$
is a CCU and that these individual CCUs collectively ensure that $U$
is also a CCU. Lemma~\ref{lem:ccuafibadj} shows that a fibred CCU
above $r$ guarantees exactly this --- no more, no less.

Definition~\ref{def:ccua} straightforwardly extends to Lawvere
categories as follows:

\begin{definition}\label{def:lwfiba}
  Let $U:\E\to\B$ and $r:\B\to\A$ be fibrations. We say that $U$ is a
  {\em fibred Lawvere category above $r$} if $U$ is a fibred CCU above $r$ 
  and $U$ is a bifibration.
\end{definition}

\noindent
The next two corollaries are immediate.

\begin{corollary}\label{cor:lwfaadj}
  Let $U:\E\to\B$ and $r:\B\to\A$ be fibrations. Then $U$ is a fibred
  Lawvere category above $r$ iff $U$ is a Lawvere category and, for
  every $a$ in $\A$, $U_a : \E_a \to \B_a$ is a Lawvere category whose
  unit and comprehension are given by the restrictions of the unit and
  comprehension, respectively, of $U$ to $\E_a$.
\end{corollary}

\begin{corollary}\label{cor:iind}
  Let $U:\E\to\B$ be a fibred Lawvere category above $r:\B\to\A$. For
  any object $a$ of $\A$ and functor $F:\B_a\to\B_a$, any
  $K_a$-preserving lifting $\hat{F}:\E_a\to\E_a$ of $F$ defines a
  sound induction rule for $\mu F$. In particular, the canonical
  $K_a$-preserving lifting from Section~\ref{sec:ind} defines a sound
  induction rule for $\mu F$.
\end{corollary}

Our first example shows that fibred induction is applicable in
situations in which indexed induction is not.

\begin{example}
Consider the mutually recursive data type

\vspace*{-0.05in}

\[\proofrule{}
             {zero : evens }
\hspace{1cm}
\proofrule{ n : odds}
          {evenSucc \; n : evens}
\hspace{1cm}
\proofrule{ n : evens }
          {oddSucc \; n : odds }
\]

\vspace*{0.1in}

\noindent
If we model types in a category $\B$, then we can model the
$2$-indexed data type of $odds$ and $evens$ using the initial algebra
of the functor $F : \B^2 \ra \B^2$ defined by $F(E,O) = (O+1,E)$.
However, we may wish to index data types by sets other than $2$.
%We therefore generalise the families fibration to the fibration
% $U:Fam(\B) \rightarrow \Set$, where an object of $\Fam(\B)$ is a pair
% $(X,f)$ with $X$ a set and $f$ a function from $X$ into the objects
% of $\B$. A morphism $(X,f) \rightarrow (X',f')$ consists of a pair
% $(k,k')$, where $k:X \rightarrow X'$ and $k'$ assigns to every $x:X$
% a morphism from $fx$ to $f'(kx)$ in $\B$.  
The codomain fibration $cod:Fam(Set)^\ra \ra Fam(Set)$ defines a
fibred Lawvere category over the families fibration.  We therefore
have the following induction rule in the families fibration for any
predicates $P:evens \ra \Set$ and $Q: odds \ra \Set$:

\[\begin{array}{l}
P(zero) \rightarrow\\
(\Pi n : odds. \,Q(n) \rightarrow P(evenSucc \; n)) \ra \\
(\Pi n : evens. \,P(n) \rightarrow Q(oddSucc\; n)) \ra \\
(\Pi n : evens. \,P(n)) \times (\Pi n : odds. Q(n))
\end{array}
\]
\end{example}

\vspace*{0.1in}

We can also see the induction rule of Theorem~\ref{thm:inda} as an
instance of fibred induction:

\begin{example}\label{ex:sec2}
  Let $U:\E \ra \B$ be a Lawvere category. If $1$ is the category with
  one object and one morphism, then $U$ is a fibred Lawvere category
  above the fibration $r :\B \ra 1$. Moreover, the treatment of
  induction from Section~\ref{sec:ind} is equivalent to the treatment
  of induction for this fibred Lawvere category.
\end{example}

And we can see indexed induction as an instance of fibred induction:

\begin{lemma}\label{lem:canlwfiba}
  Let $U:\E\to\B$ be a Lawvere category. The fibration
  $q:\E'\to\B^\to$ obtained by the change of base
  \[\xymatrix{
    \ar[d]_q \E' \ar[r] \lpbc & \E \ar[d]^U\\ \B^\to \ar[r]_{dom} & \B
  }\] is a fibred Lawvere category above the codomain fibration and,
  for any $I$ in $\B$, $q_I = U/I$.
\end{lemma}
\begin{proof}
  Consider the following setting:
  \[\xymatrix{
    \ar[d]_q \E' \ar[r] \lpbc & \E \ar[d]^U\\ \ar[d]_{cod} \B^\to
    \ar[r]_{dom} & \B\\ \B }\] By Lemma~\ref{lem:coblwfib}, $q$ is a
  Lawvere category because it arises by change of base along the
  fibration $dom$. Moreover, for any object $I$ of $\B$, the fibration
  $q_I$ can be obtained by the change of base
  \[\xymatrix{
    \ar[d]_{q_I} \E_I \ar[r] \lpbc & \E'\ar[d]^q \\
    \B_I = \B/I \ar[r]_{i_I} & \B^\to
  }\]

\vspace*{0.05in}

\noindent
  where $i_I$ is the inclusion functor. Thus $q_I$ arises as the
  pullback of $U$ along the composition of $dom:\B^\to \to \B$ and
  $i_I:\B/I \ra B^\to$. But this composition is simply $dom:\B/I \ra
  \B$, so it is clearly a fibration. Thus, $q_I$ is a Lawvere
  category, and $q$ is itself a fibred Lawvere category above the
  codomain fibration. Finally, $q_I = U/I$ by construction.
\end{proof}

\section{Fibred Coinduction}\label{sec:fibcoind}

In this section we extend the methodology of Section~\ref{sec:fibind}
to give sound coinduction rules for coinductive indexed types in the
case when the indexing is not modelled by slice categories. As in
Section~\ref{sec:fibind}, we consider a fibration $r:\B \ra \A$, where
we think of the objects of $\B$ as being indexed by the objects of
$\A$, and a fibration $U : \E \to \B$ that we think of as a logic over
$\B$. Our aim is to derive sound coinduction rules for final
coalgebras of functors $F:\B_a \ra \B_a$, where $a$ is any object of
$\A$.

Our experience from Section~\ref{sec:coind} suggests that a minimal
requirement for deriving a sound coinduction rule for a functor
$F:\B_a \ra \B_a$ is that the fibration $U_a$ is a QCE.  As in
Section~\ref{sec:fibind}, we want to highlight the uniformity
connecting the different fibrations $U_a$ but, unfortunately,
requiring that each fibration $U_a$ is a QCE does not automatically
imply that $U$ is a QCE. On the other hand, if we define a
\emph{(full) section} of a functor $F:\C\to\D$ to be a (resp., full
and faithful) functor $E:\D \ra \C$ such that $FE=id_\D$ then, for
fibrations $U:\E\to\B$ and $r:\B\to\A$, a (full) section $E:\B\to\E$
of $U$ straightforwardly restricts to a (resp., full) section
$E_a:\B_a\to\E_a$ of $U_a$ for any object $a$ of $\A$. Then, by
contrast with the situation in the inductive case, requiring that each
fibration $U_a$ is a QCE with section $E_a$ actually does ensure that
$U$ is a QCE with section $E$, provided $E$ preserves cartesian
morphisms. Indeed, observing that the notion of a cartesian morphism
and the notion of a fibre both make sense for arbitrary functors
whether or not they are fibrations, and extending our notation for
fibres of fibrations to fibres of functors, we have the following:

\begin{lemma}\label{lem:weakfibadj}
  Let $r:\B\to\A$ be a fibration, and let $q:\E\to\A$ and $e:\B\to\E$
  be functors such that $qe = r$. The functor $e$ has a left adjoint
  $Q:\E\to\B$ with vertical unit (or, equivalently, counit) iff $e$
  preserves cartesian morphisms and, for each object $a$ in $\A$, the
  restriction $e_a:\B_a\to\E_a$ of $e$ has a left adjoint $Q_a$.
\end{lemma}
\begin{proof}
  Suppose $e$ preserves cartesian morphisms, let $Q_a:\E_a\to\B_a$ be
  a collection of left adjoints to the restrictions $e_a:\B_a\to\E_a$
  of $e$, and let $\eta^a$ be the unit of $Q_a \dashv e_a$.  We will
  prove that, for each $a$ in $\A$ and $R$ in $\E_a$, the morphism
  $(\eta^a)_R:R\to eQ_a R$ is universal from $R$ to $e$ (and not just to
  $e_a$). By part (ii) of Theorem~2 of Chapter 4 of~\cite{mac71}, this
  gives an adjunction $Q \dashv e$. The unit of this adjunction is
  vertical because it comprises the various units $\eta^a$.

  To this end, consider a morphism $l:R\to e Y$ in $E$ above $h:a\to
  b$ in $\A$.  Then $Y$ is above $b$, and so there is a cartesian
  morphism $h^\S_Y:h^*Y \ra Y$ above $h$ with respect to $r$. Because
  $e$ preserves cartesian morphisms, we know that $e(h^\S_Y)$ is
  cartesian above $h$ with respect to $q$. Thus $l = e(h^\S_Y) u$ for
  a unique vertical morphism $u:R\to e(h^*Y)$ with respect to
  $q$. Now, since $u$ is in $\E_a$, we can use the universal property
  of $(\eta^a)_R$ to deduce a unique morphism $g:Q_aR\to h^*Y$ in $\B_a$
  such that $u = e(g)\eta_R$. Therefore, we have a unique morphism
  $f=h^\S_Y g$ such that $l = e(f) \eta_R$.

  Conversely, suppose $Q$ is left adjoint to $e$ with vertical unit.
  Then $e$ preserves cartesian morphisms by
  Lemma~\ref{lem:adjandcart}, and the adjunction $Q \dashv e$
  restricts to adjunctions $Q_a \dashv e_a$ because the unit of $Q
  \dashv e$ is vertical and $qe = r$.
\end{proof}

We can now give the central definitions we need to state our sound
coinduction rules for coinductive indexed types.

\begin{definition}\label{def:qcea}
  Let $U:\E\to\B$ and $r:\B\to\A$ be two fibrations, and let
  $E:\B\to\E$ a full section of $U$. We say that $U$ is a \emph{QCE
    above $r$} if $E$ has a left adjoint $Q:\E\to\B$ and the
  adjunction $Q \dashv E$ is fibred above $\A$:
  \[\xymatrix{\E\ar[dr]_{r U} \ar@/_/[rr]_Q
    \ar@{}[rr]|\top & & \ar@/_/[ll]_E \ar[ld]^r \B\\ & \A & }\] 

\vspace*{0.1in}

\noindent
Note that, in this case, both $E$ and $Q$ are necessarily fibred.  A
{\em weak QCE above $r$} is similar to a QCE above $r$, except that
the left adjoint to $E$ need not be fibred (although $E$ itself must
still be).
\end{definition}
\noindent

With this definition in place, we have the following analogue of
Corollary~\ref{cor:lwfaadj}: 

\begin{lemma}
Let $U:\E \ra \B$ and $r:\B \ra \A$ be fibrations. Then $U$ is a weak
QCE above $r$ iff $U$ is a QCE and, for any object $a$ of $\A$,
$U_a:\E_a \ra \B_a$ is a QCE whose full section and quotient functors
are given by the restrictions of the full section and quotient
functors, respectively, of $U$ to $\E_a$.
\end{lemma}

\begin{proof}
If $U$ is a weak QCE above $r$, then the fact that $U$ and the
fibrations $U_a$ are QCEs is straightforward. For the other direction,
we observe that the unit (equivalently, counit) of $Q \vdash E$ is
vertical. Lemma~\ref{lem:adjandcart} therefore guarantees that $E$
preserves cartesian morphisms, i.e., is fibred.
\end{proof}

\begin{corollary}\label{cor:icoind}
  Let $U:\E\to\B$ be a weak QCE above $r:\B\to\A$. For any object $a$
  of $\A$ and functor $F:\B_a\to\B_a$, any $E_a$-preserving lifting
  $\cech{F}:\E_a\to\E_a$ of $F$ defines a sound coinduction rule for
  $\nu F$. In particular, the canonical $E_a$-preserving lifting from
  Section~\ref{sec:coind} defines a sound coinduction rule for $\nu
  F$.
\end{corollary}

We can see Theorem~\ref{thm:coinda} as a special case of 
Corollary~\ref{cor:icoind}.

\begin{example}\label{ex:sec3}
  Let $U:\E \ra \B$ be a relational QCE. If $1$ is the category with
  one object and one morphism, then $U$ is a weak QCE above the
  fibration $r :\B \ra 1$. Moreover, the treatment of coinduction from
  Section~\ref{sec:coind} is equivalent to the treatment of
  coinduction for this weak QCE above $r$.
\end{example}

\noindent
We can also see Theorem~\ref{thm:icoinda} as a special case of
Corollary~\ref{cor:icoind}. This entails constructing, from the data
assumed in Lemma~\ref{lem:lift-qce}, a weak QCE above the codomain
fibration $\mathit{cod}$.  To do this, we first define an analogue of
a relational QCE in the setting where we are working above a fibration
$r$. We have the following definition:

\begin{definition}\label{def:relqceabover}
  Let $U:\E\to\B$ be a bifibration with truth functor $K:\B\to\E$, let
  $r:\B\to\A$ be a fibration, and assume that $r$ has fibred cartesian
  products, i.e., products in the fibres that are preserved by
  reindexing. Let $\Delta_r:r\to r$ be the fibred diagonal functor
  mapping each object $X$ in $\B_a$ to the product with itself in
  $\B_a$.  Then, the {\em relations fibration above $r$} is defined to
  be the fibration $Rel_r(U):Rel_r(\E)\to \B$ above $r$ that is
  obtained by change of base of $U$ along $\Delta_r$. If
  $\delta_r:Id_\B\to \Delta_r$ is the diagonal natural transformation
  for $\Delta_r$, then the {\em equality functor for $U$ above $r$} is
  defined to be the functor $Eq_r:\B\to Rel_r(\E)$ that maps an object
  $X$ of $\B$ to $\Sigma_{\delta_r}K X$.  Furthermore, if $Eq_r$ has a
  left adjoint $Q_r$, then $Q_r$ is called the {\em quotient functor
    for $U$ above $r$}.  A {\em relational QCE above $r$} is a QCE
  above $r$ obtained via this construction. A {\em weak relational QCE
    above $r$} is similar to a relational QCE above $r$, except that
  the left adjoint to the fibred equality functor need not be fibred.
\end{definition}

The main difficulty in constructing a weak relational QCE $U$ above a
fibration $r$ is proving that the equality functor for $U$ is
fibred. If $\B$ has pullbacks and $r : \B^\to \to \B$ is the codomain
fibration, then $r$ has fibred products given by pullbacks.  In this
case, we write $\Delta^\ra:\B^\ra \ra \B^\ra$ for the functor mapping
an object $f$ in $(\B^\ra)_I$ to the product $f^2$ of $f$ with itself
in the fibre $(\B^\ra)_I$. We denote the diagonal natural
transformation for $\Delta^\to$ by $\delta^\ra:Id_{\B^\ra} \ra
\Delta^\ra$.

\begin{lemma}\label{lem:carts}
Let $U:\E \ra \B$ be a bifibration with a truth functor, and suppose
$U$ satisfies the Beck-Chevalley condition. Furthermore, assume that
$\B$ has products and pullbacks. Let $U':\E' \ra \B^\ra$ be obtained
from $U$ by change of base along the fibration $\mathit{dom}$, and let
$Rel(U'):Rel(\E') \ra \B^\ra$ be obtained from $U'$ by change of base
along $\Delta^\ra : \B^\to \to \B^\to$:
\[\xymatrix{ Rel(\E') \ar[d]_{Rel(U')} \ar[r] \lpbc &\ar[d]_{U'} \E' \ar[r] \lpbc
    & \E \ar[d]^U\\ \B^\to \ar[r]_{\Delta^\to} & \B^\to \ar[r]_{dom} &
  \B\\ }\] 
Finally, let $\Sigma_{\delta^\ra}:\E' \ra Rel(\E')$ be the functor
that maps an object $(h:UP \ra I, P)$ of $\E'$ to the object
$(h^2,\Sigma_{\delta^\ra_h} P)$ of $Rel(\E')$. Then
$\Sigma_{\delta^\ra}$ is a fibred functor from $cod\circ U'$ to
$cod\circ Rel(U')$.
\end{lemma}

\begin{proof}
Let $(h : UP \to I,P)$ be an object of $\E'$, let $f:J \ra I$ be a
morphism of $\B$, and consider the pullback $(X, \psi:X \ra J, \phi:X
\ra UP)$ of $h$ along $f$.  Then the cartesian morphism above $f$ with
codomain $(h,P)$ is the morphism $(f, \phi^\S_P) : (\psi,\phi^*P) \ra
(h, P)$ in $\E'$. We must show that the morphism $\Sigma_{\delta^\ra}
(f, \phi^\S_P) : \Sigma_{\delta^\to_\psi}\phi^*P \to
\Sigma_{\delta^\to_h}P$ is cartesian.

We begin by considering the morphism $(f,\phi):\psi \ra h$ in
$\B^\ra$. The fact that $(X,\psi,\phi)$ is a pullback means that
$(f,\phi)$ is cartesian. Because $\Delta^\ra$ is fibred, we know that
$\Delta^\ra (f,\phi)$ is cartesian. Thus, if $\Delta^\to (f,\phi) =
(f, \alpha)$, then $(X_f X, \psi^2, \alpha)$ is the pullback of $h^2$
along $f$.
From part b of Exercise 8 on page 72 of~\cite{mac71},
together with the facts that $(X, \psi, \phi)$ and $(X_f X, \psi^2,
\alpha)$ are pullbacks, we have that $(X,\delta^\ra_\psi, \phi)$ is
the pullback of $\delta^\ra_h$ along $\alpha$. The Beck-Chevalley
condition thus ensures that $\Sigma_{\delta^\ra_\psi} \phi^*P$ is
isomorphic to $\alpha^*\Sigma_{\delta^\ra_h}P$.  Letting $Q$ stand for
$\Sigma_{\delta^\ra_h} P$, we therefore have that $\Sigma_{\delta^\ra}
(f, \phi^\S_P)$ is $(\alpha^\S_Q, (f,\alpha))$. Since $\alpha^\S_Q,$ is
cartesian with respect to $U$ by definition, and since $(f,\alpha)$ is
cartesian with respect to $\mathit{cod}$, we have that
$\Sigma_{\delta^\ra} (f,\phi^\S_P)$ is cartesian with respect to
$\mathit{cod}\circ Rel(U')$, as required.
\end{proof}

\begin{lemma}\label{lem:wqcer}
  Let $U : \E \to \B$ be a bifibration such that the Beck-Chevalley
  condition holds and $\B$ has pullbacks.  Let $Rel(U):Rel(\E)\to\B$
  be the relational QCE derived from $U$, with equality functor
  $Eq:\B\to Rel(\E)$ and quotient functor $Q:Rel(\E)\to\B$.  Then the
  bifibration $Rel(U'):Rel(\E')\to\B^\to$ obtained from the following
  change of base
  \[\xymatrix{
    \ar[d]_{Rel(U')} Rel(\E') \ar[r] \lpbc & \E' \ar[d]^{U'}\\
    \B^\to \ar[r]_{\Delta^\to} & \B^\to }\]
  is a weak QCE above $cod$. In addition, for any $I$ in $\B$,
  $Rel(U')_I \cong Rel(U/I)$.
\end{lemma}
\begin{proof}
  We have the following situation:
  \[\xymatrix{
    Rel(\E') \ar[d]_{Rel(U')} \ar[r] \lpbc &\ar[d]_{U'} \E' \ar[r] \lpbc
    & \E \ar[d]^U\\  
    \B^\to \ar[rd]_{cod} \ar[r]_{\Delta^\to} & \ar[d]_{cod} \B^\to
    \ar[r]_{dom} & \B\\ 
    &\B
  }\]
  Let $I$ be an object of $\B$. First, by the same reasoning as in the
  proof of Lemma~\ref{lem:canlwfiba} we have that $U'_I = U/I$. Now,
  note that if $i_I:\B/I\to\B^\to$ is the inclusion functor, then
  $\Delta^\to i_I = \Delta/I$. Thus, $Rel(U')_I$ is obtained by change
  of base of $U/I$ along the functor $\Delta/I$, i.e., the fibrations
  $Rel(U')_I$ and $Rel(U/I)$ coincide.  Moreover, the equality functor
  $Eq^\to:\B^\to\to Rel(\E')$ is defined by
  $\Sigma_{\delta^\to}K^\to$, where $K^\to$ is the truth functor for
  $U'$. Then $Eq^\to$ restricts to the equality functor $Eq_{U/I}$
  since $\delta/I$ is the restriction of $\delta^\to$ to the
  corresponding fibres. This ensures that the equality functors for
  $Rel(U')_I$ and $Rel(U/I)$ coincide.  Finally, because adjoints are
  defined up to isomorphism, the quotient functors for $Rel(U')_I$ and
  $Rel(U/I)$ coincide. Putting this all together, we have that
  $Rel(U')_I$ and $Rel(U/I)$ are in fact the same QCE.  Now, each
  restriction $Eq^\to_a$ of $Eq^\to$ has a left adjoint $Q^\to_a$.
  Moreover, $Eq^\to$ preserves cartesian morphisms because $K^\to$
  preserves cartesian morphisms by construction, and
  Lemma~\ref{lem:carts} ensures that $\Sigma_{\delta^\to}$ preserves
  cartesian morphisms. By Lemma~\ref{lem:weakfibadj}, we have that
  $Q^\to$ is a left adjoint to $Eq^\to$, and thus that $Rel(U')$ is a
  weak QCE above $\mathit{cod}$.
\end{proof}

The coinduction rule for the mutually recursive data type of {\em
  odds} and {\em evens} in the families fibration shows that fibred
coinduction is applicable in situations where indexed coinduction is
not.

\section{Conclusions, Related Work, and Future Work}\label{sec:conc} 

In this paper, we have extended the fibrational approach to induction
and coinduction pioneered by Hermida and Jacobs, and further developed
by the current authors, in three key directions: we have given sound
coinduction rules for all (unindexed) coinductive types, and we have
extended our results from the unindexed setting to the indexed one to
derive sound induction and coinduction rules for all inductive and
coinductive indexed types. We derived our rules for indexed types
first in the case when indexing is modelled by the codomain fibration,
and then in the case when it is modelled by an arbitrary fibration.

The work of Hermida and Jacobs is most closely related to ours, but
there is, of course, a large body of work on induction and coinduction
in a broader setting. In dependent type theory, for example, data
types are usually presented along with elimination rules that are
exactly induction rules. Along these lines,~\cite{pm89} has heavily
influenced the development of induction in Coq. Another important
strand of related work concerns inductive families and their induction
rules~\cite{dyb94}. On the coinductive side, papers such
as~\cite{am89,rut00,tr98} have had immense impact in bringing
bisimulation into the mainstream of theoretical computer science.

There are several directions for future work. First, we would like to
explore more applications of the results in Sections~\ref{sec:fibind}
and~\ref{sec:fibcoind}.  More generally, we would like to exploit the
predictive power of our theory to provide induction and coinduction
rules for advanced data types --- such as inductive recursive types
--- for which these rules are not discernible by sheer intuition. In
such circumstances, our generic fibrational approach should provide
rules whose use is justified by their soundness proofs. In a different
direction, we would like to see our induction and coinduction rules
for advanced data types incorporated into implementations such as Agda
and Coq.

\vspace*{0.1in}

\noindent
{\bf Acknowledgement} We thank the reviewers for their helpful
comments and suggestions.

\bibliographystyle{plain}

\end{document}